\documentclass[pdflatex,sn-basic]{sn-jnl}
\usepackage{graphicx}%
\usepackage{multirow}%
\usepackage{amsmath,amssymb,amsfonts}%
\usepackage{todonotes}
\usepackage{amsthm}%
\usepackage{mathrsfs}%
\usepackage[title]{appendix}%

\usepackage{textcomp}%
\usepackage{manyfoot}%
\usepackage{booktabs}%
\usepackage{algorithm}%
\usepackage{algorithmicx}%
\usepackage{algpseudocode}%
\usepackage{listings}%
\usepackage{bm}
\usepackage{braket}
\usepackage{centernot}
\usepackage{bbold}
\usepackage[T1]{fontenc}
\usepackage{mwe}

\usepackage[dvipsnames]{xcolor}

\usepackage{soul}

\usepackage{tikz}
\usepackage{quantikz}
\usepackage{circuitikz}
\usetikzlibrary{arrows,decorations.pathreplacing}
\usetikzlibrary{decorations.pathreplacing}
\usetikzlibrary{shapes.geometric, arrows.meta, positioning}

\theoremstyle{thmstyleone}%
\newtheorem{theorem}{Theorem}
\newtheorem{lemma}[theorem]{Lemma}
\newtheorem{remark}[theorem]{Remark}

\begin{document}
\begin{titlepage}
\begin{center}
\LARGE \text{On Quantum Perceptron Learning via Quantum Search}
\end{center}
\vspace{0.2cm}
\large\textbf{Author information}
\\
\normalsize
Xiaoyu~Sun$^{1*}$, Mathieu~Roget$ ^1$, Giuseppe Di~Molfetta$^{1,2}$, and Hachem~Kadri$^1$
\\
$^{1}$Aix-Marseille Universit\'e, CNRS, LIS, Marseille, France
$^{2}$ Institut Universitaire de France, Paris, France
\\
$^{*}$Corresponding author, e-mail: <xiaoyu.sun@lis-lab.fr>
\\\\
\large\textbf{Abstract}
\normalsize
\\
With the growing interest in quantum machine learning, the perceptron, a fundamental building block in traditional machine learning, has emerged as a valuable model for exploring the potential of quantum algorithms. In this work, we make two principal contributions. First, we revisit the \emph{quantum version space perceptron} algorithm proposed by \cite{kapoor2016quantum}, by identifying and correcting a flawed complexity assumption. We show that the query complexity of the algorithm is dimension-dependent, which has significant implications for its behaviour in high-dimensional regimes under worst-case scenarios.

Second, we propose and analyse two \emph{quantum-enhanced} cutting-plane algorithms for perceptron learning. Specifically, we leverage established quantum subroutines such as \emph{Grover's search} and \emph{quantum walk search}, and provide detailed algorithmic constructions together with query and arithmetic complexity analyses.
Our results establish improved complexity bounds under an idealised implementation framework and noise-free quantum computational models, offering insights into the trade-offs between margin dependence, dimensional dependence, and quantum resources. These findings provide a refined understanding of quantum perceptron models and their theoretical computational complexity properties.
\\\\
\large\textbf{Statements and Declarations}
\normalsize
\begin{itemize}
\item Competing interests:
The authors have no competing interests to declare that are relevant to the content of this article.
\end{itemize}
\vspace{0.2cm}
\large\textbf{Acknowledgements}
\normalsize
\\
This work is supported by the Plan France 2030 through the PEPR integrated project EPiQ ANR-22-PETQ-0007; by the ANR JCJC DisQC ANR-22-CE47-0002-01; by the QuanTEdu-France ANR-22-CMAS-0001, and as part of the Initiative d’Excellence d’Aix-Marseille Université—A*MIDEX AMX-21-RID-011.
\\\\

\end{titlepage}

\section{Introduction}\label{sec1}

With the great development of computational power and algorithms, machine learning has gained more and more attention in the modern stage, with the increasing demand to learn from \emph{big data} of high feature space dimension $D$, which greatly motivates the use of quantum computing in machine learning. This has led to the emergence of the interdisciplinary field of \emph{quantum machine learning} ~\citep{biamonte2017quantum,schuld2015introduction}. 

Originating from the suggestion of \cite{feynman1982simulating} to use quantum hardware to simulate quantum dynamics with exponential overhead, quantum computing integrates knowledge from quantum algorithms, quantum information processing and experimental quantum systems. One of its central features is the superposition principle, which allows quantum states to represent high-dimensional data compactly. Additionally, quantum computing supports inherent parallelism, enabling the use of quantum linear algebra subroutines such as the quantum Fourier transform~\citep{nielsen2010quantum} and the HHL algorithm~\citep{harrow2009quantum} that operate in $\mathrm{polylog} (D)$ time, or some others subroutines with quadratic speedup, such as Grover’s search ~\citep{grover1996fast}  and its extension quantum counting algorithm~\citep{brassard1998quantum}.

Despite these theoretical advantages, several practical challenges remain. Efficient data encoding into quantum states, especially via amplitude encoding, remains an open problem, as suggested by \cite{aaronson2015read}, particularly in the absence of scalable quantum random access memory (QRAM). Moreover, many quantum linear algebra routines are highly sensitive to noise and are currently limited to specific physical platforms. Nevertheless, the exploration of quantum advantages remains theoretically valuable, as it may shed light on the boundaries of classical complexity. Moreover, from a practical standpoint, early experimental results and quantum circuit simulations have already demonstrated promising quantum advantages in some machine-learning tasks~\citep{schuld2020circuit,saggio2021experimental}.

\subsection{Motivation and Contributions}
In supervised machine learning, classification requires algorithms to give the correct label to unseen (test) data based on knowledge recognised from labelled training data.  
The simplest task in classification is likely binary classification on linearly separable data, which can be solved by the \emph{online perceptron algorithm}~\citep{rosenblatt1958perceptron}, with a tight update bound that only depends on the geometric margin of data, proved by \cite{novikoff1962convergence,block1962perceptron}. 
Although other classification methods, such as \emph{support vector machine} (SVM), which relies on advanced optimisation techniques to find optimal separating hyperplanes, the online perceptron algorithm uses a simple, iterative update rule and serves as the basic building unit of \emph{multilayer perceptron} (MLP), i.e. the single-layer perceptron, which is also a minimal \emph{feedforward neural network} (FNN). 
Due to its elementariness and simplicity, it is considered an ideal theoretical lens to study statistical and computational trade-offs of quantum machine-learning algorithms, as pointed out by \cite{kapoor2016quantum}. However, prior quantum perceptron models, particularly the quantum version space perceptron~(QVSP) algorithm introduced by \cite{kapoor2016quantum}, rest on an assumption about margin-based sampling (in their proof of Theorem~2), which is theoretically flawed in higher feature space dimensions $D>2$. This raises the critical question of exploring high-dimensional available quantum perceptron algorithms.

In this work, we make two principal contributions. First, we revisit the QVSP algorithm proposed by \cite{kapoor2016quantum} and prove that the probability of sampling a $D$-dimensional hyperplane from a standard normal distribution that perfectly classifies the data scales as $\Omega(\gamma^D)$ in the worst case, rather than $\Theta(\gamma)$, where $\gamma$ is the geometric margin. This correction is based on a critical assumption of \cite{kapoor2016quantum} and extends to subsequent works in the quantum perceptron literature~\citep{roget2022quantum,liao2024quadratic}. This finding prompts a question about whether quantum computing can address margin dependency and what quantum perceptron algorithms should be applied in high-dimensional feature spaces.

Our second contribution answers this question by proposing two {quantum-enhanced} cutting-plane (CP) algorithms for perceptron learning. Specifically, a \emph{hybrid} quantum-classical \emph{cutting-plane random walk} (HCP-RW)  and a \emph{fully-quantum} 
\emph{cutting-plane quantum walk} (QCP-QW). Compared to the classical {online perceptron algorithm}~\citep{rosenblatt1958perceptron}, these algorithms exhibit improved margin dependence arising from linear programs (LPs), which enjoys $O^*(D\log {(1/\gamma)})$\footnote{We use the asterisk notation $O^*(\cdot)$ to suppress the error parameter and factors like $\log(D)$, $\log\log(1/\gamma)$ and $\log\log(N)$ throughout the paper.} number of updates. Meanwhile, within idealised oracle-based quantum models, they also yield a sublinear dependence of $O(\sqrt{N})$ in the number of training examples, analogous to the \emph{online quantum perceptron} (OQP) algorithm~\citep{kapoor2016quantum}. Furthermore, the QCP-QW algorithm provides improved complexity scaling by $O^*(D^{1.5})$ speedup in both the number of queries\footnote{Throughout the paper, the terms classical and quantum query, as well as query complexity, follow the definitions given in Section~\ref{data_query}.} and arithmetic operations compared with the HCP-RW.  
Together, the two contributions provide a refined understanding of the behaviour of earlier quantum perceptron models and introduce new quantum-enhanced perceptron learning algorithms grounded in convex optimisation.

While both the HCP-RW and QCP-QW algorithms focus on theoretical computational complexity analysis, we acknowledge that, the reported speedups should be viewed as characterising the asymptotic potential of quantum perceptron models under idealised setup\footnote{Note that throughout the paper, we phrase the terms of "query complexity", "computational complexity" and "arithmetic complexity" etc., as the complexities under idealised, noise-free quantum computational model, which should not be interpreted as real-device performances.} rather than guarantees of real-time performance, as the implementation on near-term noisy intermediate-scale quantum (NISQ) or early fault-tolerant devices requires error correction or mitigation schemes, which are beyond the scope of the present work. Therefore, we provide detailed pseudocode, complexity analyses, and algorithmic flow for QCP-QW and defer simulations, small-scale tests, and fault-tolerant schemes to future work.

\subsection{Outline}
The rest of our paper is organised as follows. In Section~\ref{related_work}, we first review the connections between \emph{perceptron learning}, \emph{version space} and \emph{feasibility problem}. Then we re-examine some known quantum perceptron models~\citep{kapoor2016quantum,roget2022quantum,liao2024quadratic} and provide the formal definitions of the \emph{classical/quantum queries} related to this work. Section~\ref{MC_sampling} corrects an error in the QVSP algorithm and provides a simulation validating this correction. In Section~\ref{Q_enh},  we deliberate on quantum-enhanced perceptron learning algorithms by first explaining the \emph{quantum weak online learning} framework, and then explicating two CP-based algorithms. The quantum-classical hybrid algorithm HCP-RW, and a fully-quantum algorithm QCP-QW that improved on the HCP-RW algorithm, elaborating on how to use {quantum walk search}~\citep{magniez2007search} to speed up sampling. Section~\ref{conclu}  concludes the paper and points out future directions.

\section{Background and Preliminaries}\label{related_work}
\subsection{Perceptron, Version Space and Feasibility Problem}\label{relations}
Here we follow a deterministic setup (in comparison with the stochastic setup) where the labels $y_i=f(\bm x_i)$ are uniquely determined by some measurable function $f$ with probability one. A sample $\mathcal{S} = (\bm x_1, \ldots, \bm x_N  ) \sim \mathcal{D}^N$ is a sequence of instances drawn i.i.d. according to a distribution $\mathcal{D}$ and  we use the notation $\{(\bm x_i,y_i)_{i\in[N]}\}$ to refer to a sequence of $N$ pairs of training dataset $\{(\bm x_1,y_1),\ldots,(\bm x_N,y_N)\}$. A hyperplane $\bm w'^T\bm x +b=0$ in the $ D$-dimensional space can be defined by a normal vector $\bm w'\in \mathbb{R}^D$ and an intercept $b$. Here, we only consider the case where $b=0$, which can easily be achieved by redefining $\bm w^T=(\bm w'^T,b)$ in the enlarged $D+1$ dimensional space. In the separable case, there exists a hyperplane $\bm w$ separating two classes of data with their correspondence labels perfectly, i.e., $\bm w^T \bm x_i y_i >0$, $\forall i\in [N]$.\footnote{In here and the rest of the paper, we refer to a hyperplane $\bm w^T\bm x=0$ defined by its normal vector $\bm w$ as a hyperplane $\bm w$ for short.}   Moreover, when there is a \emph{maximal margin hyperplane} $\bm u$ separating the two classes with a maximal \emph{geometric margin} $\gamma>0$ such that $\bm u^T\bm x_i y_i \geq \gamma$ for all $i\in [N]$, the vector $\bm u$ is known as the SVM solution.

Now, given $N$ normalized linearly separable training examples $ \{\bm x_1, ., \bm x_N  \} \in \mathbb{R}^D$, $||\bm x_i||= 1$,\footnote{We use $||\cdot||$ to represent the $\ell_2$ norm throughout the paper.}  with corresponding labels $ \{y_1, . ., y_N  \}, y_i \in\{+1,-1\}$. The \emph{perceptron itself} is a weight vector (hyperplane) $\bm w$  that follows a prediction rule $\hat{y}_i=\text{sgn}(\bm w^T \bm x_i)$ for every unlabelled instance $\bm x_i$.\footnote{We take the \text{sgn} function such that $\text{sgn}(x)=1$, if $x> 0$; and $\text{sgn}(x)=-1$, if $x\leq0$.}
\emph{Perceptron learning} aims to find a hyperplane $\bm w ^T \bm x_i  y_i >0$, $\forall i \in [N]$, such that all training examples are classified correctly.
The online perceptron algorithm~\citep{rosenblatt1958perceptron} is a foundational online learning method introduced in the early 1960s. At each update step, the model maintains a weight vector $\bm{w}_t$ and updates it as $\bm{w}_{t+1} \leftarrow \bm{w}_t + y_i' \bm{x}_i'$, where $(\bm{x}_i', y_i')$ is a training example misclassified by $\bm{w}_t$. 

It is well known that if the training dataset is linearly separable with margin $\gamma$, the online perceptron algorithm converges after at most $O(1/\gamma^2)$ updates, resulting in a linear classifier that correctly classifies all $N$ training points~\citep{novikoff1962convergence}. Geometrically, the solution lies in the version space, defined as $\mathcal{VS}:=\{\bm w \mid \bm w ^T \bm x_i  y_i >0, \forall i \in [N]\}$. As only directions of $\bm w$ matter for classification, we can limit the length of $||\bm w||\leq 1$, the $\mathcal{VS}$
is then a bounded polyhedron (a polytope) satisfying $N$ half-space constraints, also known as a convex \emph{feasible set} in the context of linear programming.  Notably, the solution to the hard-margin SVM problem corresponds to the centre of the largest inscribed sphere within $\mathcal{VS}$ as pointed out by \cite{tong2001support}. 

Notice that the solution of a classical online perceptron algorithm converges to an arbitrary point $\bm w_r\in \mathcal{VS}$, so the problem of perceptron learning can be reinterpreted as a feasibility  problem through the lens of linear programming~\citep{grotschel2012geometric,nemirovski1994efficient}.
Given this, we point out that general feasibility algorithms can also serve as primitive approaches to address the perceptron learning problem with high-dimensional data, such as the ellipsoid method~\citep{khachiyan1979polynomial,yang2009online}, and various cutting-plane methods~\citep{vaidya1996new,atkinson1995cutting,bertsimas2004solving,louche2015cutting,lee2015faster}. A detailed summary can be found in a recent paper of \cite[Table 1]{jiang2020improved}.

\subsection{Quantum Perceptron Models}\label{data_query}

Motivated by the classical perceptron algorithm, several quantum variants have been proposed to accelerate perceptron training, improve statistical efficiency, and enhance generalisation performance~\citep{kapoor2016quantum,roget2022quantum,liao2024quadratic}. 

In all the perceptron models, a \emph{classical query} on a training data pair $(\bm x'_i, y'_i)$  refers to an evaluation of a Boolean function $f_{\bm w}$, where 
\begin{align}\label{Boolean}
    f_{\bm w}(\bm x'_i y'_i)=1 \iff \bm w^T\bm x'_i y'_i\leq0  
      & \text{ (misclassified}),
\end{align}
acting on one data pair $(\bm x'_i,y'_i)$ each time. 

A \emph{quantum query} on a training data $\ket{\bm{x}'_i y'_i}$ refers to an evaluation of a quantum oracle (a unitary operation) $F_{\bm{w}}$  which acts coherently on a superposition of state vectors,
 \begin{equation}\label{online_q}
 F_{\bm w}\ket{\bm x'_iy'_i}=(-1)^{f_{\bm w} ( \bm x'_i y'_i  )}\ket{\bm x'_iy'_i},
\end{equation}
or equivalently,
\begin{equation*}
    F\ket{\bm w}|\bm x'_iy'_i\rangle \ket{0}=\ket{\bm w}|\bm x'_iy'_i\rangle\ket{0\oplus {f_{\bm w} (\bm x'_iy'_i)}}.
\end{equation*}

The critical assumption for performing such quantum queries is the availability of an efficient quantum data access scheme, as discussed in detail by \cite{kapoor2016quantum}. The scheme is commonly modelled as a quantum oracle, enabling coherent loading of a superposition over training examples.

In the OQP algorithm, assume there is an efficient quantum data access scheme, and each training pair can be encoded as a quantum state of the form $\ket{{x}_i, y_i}$. Compared to the classical online perceptron algorithm, the OQP algorithm reduces the number of classical queries from $O^*( {N} /{\gamma^2} )$ in the classical setting to $O^*( {\sqrt{N} \log(1/\gamma^2)}/{\gamma^2} )$ quantum queries by using Grover’s search to locate misclassified examples efficiently.
However, the number of perceptron updates remains $O^*(1/\gamma^2)$, the same as in the classical case. As a result, this algorithm can still be limited by poor convergence behaviour when the dataset has a very small margin $\gamma$. 

To address this issue, \cite{kapoor2016quantum} proposed the QVSP algorithm, where Grover's search is used as an inner loop to amplify the probability of finding a hyperplane in version space to at least $1/4$, and an outer loop of repeated sampling ensures the success probability is close to one. Argued by \cite{kapoor2016quantum}, this method improves the overall query complexity from $O^*(N/\gamma)$ to $O^*(N/\sqrt{\gamma})$, offering faster convergence for small-margin data.\footnote{However, we show in the next section that this argument does not hold.} Being a fully-quantum algorithm, in QVSP, a unitary operator $ G_{\mathcal{S}}$ is defined to act on an arbitrary hyperplane $\ket{\bm w}$ of the parameter space, 
 \begin{equation}\label{version_q}
     G_{\mathcal{S}}\ket{\bm w}=(-1)^{1+ (f_{\bm w}(\bm x_1,y_1)\lor f_{\bm w}(\bm x_2,y_2)\lor \cdots\lor f_{\bm w_j}(\bm x_N,y_N)  )}\ket{\bm w},
 \end{equation}
which can be achieved by using $O^*(N)$ classical queries that  a negative phase is added when $\bm w$ classifies all examples correctly.

In the subsequent works, \cite{roget2022quantum} tested the computational complexity of these two models and proposed a hybrid model with  $O^* (\sqrt{N}\log({1}/{\gamma})/\gamma )$  quantum queries and better generalisation performance;
an improved-QVSP algorithm was proposed by \cite{liao2024quadratic}, based on a quantum counting subroutine to realize the unitary operator of Eqn.\@\eqref{version_q} with $O^*(\sqrt{N})$ quantum queries, reducing the total query complexity from $O^* (N/\sqrt{\gamma}  )$ to $O^* (\sqrt{N/\gamma}  )$. However, these two papers  
are both based on the statistical argument of \cite{kapoor2016quantum} that the probability of sampling a vector $\sim\mathcal{N}(\bm 0,\mathbb{1})$ in $\mathbb{R}^D$ that lies in the version space is $\Theta(\gamma)$. In the next section, we show that this argument is not exact.

\section{Corrected Analysis of Version Space Sampling} \label{MC_sampling}
\subsection{Review of Prior Analysis}
In the QVSP algorithm proposed by \cite{kapoor2016quantum}, it is stated that the probability of sampling according to the normal distribution in $\mathbb{R}^D$ for a perfect hyperplane scales as $\Theta(\gamma)$. Then, Grover's search is used as a substitution for the Monte-Carlo sampling and provides a square-root speedup over the margin. However, in this section,  we show that this result is not exact and we give the correction in Theorem~\ref{theorem_3}. 

For the maximum margin separator $\bm u$ and a point $\bm{x}_k$ on the margin, one has $y_k \bm{u}^T  \bm {x}_k=\gamma$. When a perturbed vector $\bm w$ close to $\bm u$ that also leads to perfect classification, one has $0=\cos{(\pi/2)}<y_k \bm{w}^T \bm {x}_k=l_k<\cos{(\pi/2-2\arcsin(\gamma))}\leq 2\gamma$, $0<\gamma<1$. As $\gamma \rightarrow 0$, it can be approximated by $0< l_k< 2\gamma$. In the proof of \cite{kapoor2016quantum}, the condition $\{\bm w\mid 0<l_k<2\gamma\} $  is perceived as an equivalent for $\{\bm w \mid \bm w \in \mathcal{VS}\}$. However, we point out that this is incorrect
as the former is merely a necessary condition for the latter. By definition, $\{\bm w\mid \bm w \in \mathcal{VS}\} =\{\bm w\mid  y_i \bm w\cdot \bm x_i>0, \forall i\in [N]\}$, but  $ \{\bm w \mid  y_i \bm w\cdot \bm x_i>0, \forall i\in [N] \}\subsetneq \{\bm w\mid 0<l_k<2\gamma\} $. 
As a consequence, the statement of~\cite[Theorem 2]{kapoor2016quantum} that $\Pr[\bm w\in \mathcal{VS}]=
\Pr[ 0<l_k<2\gamma ]=\Theta(\gamma) 
$ is not true.

\subsection{Corrected Theorem and Proof} \label{theoremand proof}

\begin{lemma}\label{lemma_1}
The sufficient condition for $\bm w\in \mathcal{VS}$ in $ \mathbb{R}^D$ is that the angular distance between 
the hyperplane $\bm w$ and the maximal margin separator $\bm u $ satisfies $\measuredangle (\bm{w},\bm u)< \arcsin\gamma$.
\end{lemma}

\begin{proof}
By the triangular inequality on angular distance, it holds that
\begin{align*}
\measuredangle (\bm{w},y_i\bm{x_i}) &\leq \measuredangle (\bm{w},\bm u)+ \measuredangle (\bm u, y_i\bm{x_i}) \leq \measuredangle (\bm{w},\bm u+ \frac{\pi}{2}-\arcsin\gamma)\\
&< \arcsin\gamma + \frac{\pi}{2}-\arcsin\gamma = \frac{\pi}{2}.
\end{align*}
Therefore $\bm{w}^T \bm{x_i}y_i = ||\bm w|| \cos \measuredangle (\bm{w}, y_i\bm{x_i})> 0 $ for all $i \in[N]$, indicating $\bm{w}\in \mathcal{VS}$.  
\end{proof} 
\begin{lemma}\label{lemma_2}
The probability that sampling from  $\mathcal{N}(\bm 0,\mathbb{1})$  for a $D$-dimensional vector $ \bm{w}$ that satisfies
$\measuredangle (\bm{w},\bm u) < \arcsin\gamma $ is $ \frac{1}{2} \mathrm{I}_{\gamma^2}(\frac{D-1}{2},\frac{1}{2})$, where $\mathrm{I}_{x} (a, b)$ is the regularized incomplete beta function.
\end{lemma}

\begin{proof}
The set $\{\bm{w}\mid \measuredangle (\bm{w},\bm u)< \arcsin\gamma\} $ uniquely determine a hyperspherical sector in a $D$-dimensional unit ball $B_1$,\footnote{We use $B_R$ to refer to a $D-$dimensional ball with radius $R$ and its centre at origin.} which has colatitude angle $\arcsin{\gamma}$ and its volume is given by $V^{sector}(\arcsin\gamma)=\frac{1}{2} V(B_1)\mathrm{I}_{\gamma^2 } (\frac{D-1}{2}, \frac{1}{2} )$~\citep{li2010concise}. Since the multivariate normal distribution is isotropic, the probability of sampling a vector $\bm w$ from $\mathcal{N}(\bm 0,\mathbb{1})$ or $\mathcal{U}_{B_1}$ \footnote{ We use $\mathcal{U}_{B_1}$ to denote a uniform distribution inside of $B_1$ with probability density function defined by $f(\bm x)=1/V(B_1)$, $\forall \bm x\in B_1$. } that falls into the hyperspherical sector are the same. Consequently,
\begin{equation*}
\begin{aligned}
    \Pr [\{\bm{w}|\measuredangle (\bm{w},\bm u)< \arcsin\gamma\}  ]&=\frac{V^{sector}(\arcsin{\gamma})}{V(B_1)}= \frac{1}{2}\mathrm{I}_{\gamma^2}(\frac{D-1}{2},\frac{1}{2}).
\end{aligned}
\end{equation*}  
\end{proof}

\begin{theorem}\label{theorem_3}
Given a training dataset $\{(\bm x_i,y_i)_{i\in[N]}\}$ separated by a maximal geometric margin of $\gamma$ in $\mathbb{R}^D$, the probability of sampling a $D$-dimensional vector $\bm w\sim\mathcal{N}(\bm 0,\mathbb{1})$ which perfectly separates the data can be geometrically defined as the quotient of the volumes of version space and the $D$-dimensional unit ball $B_1$, such that

$$
\Pr[\bm w\in \mathcal{VS}]=\frac{V(\mathcal{VS})}{V(B_1)}=
\begin{cases}
\operatorname\Omega(\gamma^{D}),  & \text{when  $\gamma  \rightarrow 0$, $D\geq 2$ is a constant;} \\
\operatorname\Omega( \frac{\gamma^{D}}{\sqrt{ D}}), & \text{when $\gamma  \rightarrow 0$  and $D  \rightarrow \infty$.}
  \end{cases}
$$
\end{theorem}
\begin{proof}
We note that only directions of $\bm w$ determine classification results, and since the normal distribution $\mathcal{N}(\bm 0,\mathbb{1})$ is isotropic, the probabilities of sampling a vector $\bm w $ from $\mathcal{N}(\bm 0,\mathbb{1})$ or from $\mathcal{U}_{B_1}$ that it falls in the version space are the same. As a result, one has $\Pr[\bm w\in \mathcal{VS}]=\frac{V(\mathcal{VS})}{V(B_1)}$.
From Lemma~\ref{lemma_1}, we have $\{\bm{w}\mid\measuredangle (\bm{w},\bm u)< \arcsin\gamma\}  \subseteq  \{\bm w \mid\bm w\in \mathcal{VS}\}$, so the relationship of their probabilities satisfies $ \Pr[\bm w\in \mathcal{VS}] \geq \Pr [ \{\bm{w}\mid\measuredangle (\bm{w},\bm u)< \arcsin\gamma\}]$. By asymptotic approximation of the incomplete beta function of Lemma~\ref{lemma_2} (detailed analysis is provided in Appendix \ref{Appen_A}), when $D\geq 2$ is a constant,
\begin{equation} \label{gamma_asymp}
    \Pr [\{\bm{w}|\measuredangle (\bm{w},\bm u)< \arcsin\gamma\}  ] \sim\frac{\gamma^{D-1}}{(D-1)}/ {\mathrm{B}(\frac{D-1}{2},\frac{1}{2})} \geq \frac{\gamma^{D-1}}{\pi (D-1)}.
\end{equation}
And in the limit of  $D \rightarrow\infty$, it yields
\begin{equation} \label{d_gamma_asymp}
    \begin{aligned}
    &\Pr [\{\bm{w}|\measuredangle (\bm{w},\bm u)< \arcsin\gamma\}  ] \sim\frac{\gamma^{D-1}}{(D-1)}/ \sqrt{\frac{2\pi}{{D-1}}} 
    \sim \frac{\gamma^{D-1}}{\sqrt{2\pi (D-1)}}.
    \end{aligned}
\end{equation}
Hence, the result of Theorem~\ref{theorem_3} follows.
\end{proof}

Classically, a straightforward Monte-Carlo sampling method can be applied, such that samples $K$ hyperplanes according to $\mathcal{N}(\bm 0,\mathbb{1})$ with a failure possibility at most some $\delta$, would require
\begin{equation*}
(1-p)^K \leq \exp (-pK)\leq \delta.
\end{equation*}
Therefore, the number of samples $K= \lceil \frac{1}{p} \ln (1 / \delta)  \rceil=O (\frac{1}{\gamma^{D}}\log (1 / \delta) )$ grows exponentially with the data dimension in the limit of small margin. This exponentially large number of samples can be explained by a well-known phenomenon in sampling, dubbed the "\emph{curse of dimensionality}", describing an exponential increase in volume when adding an extra dimension. The same problem regarding exponentially small version space has also been addressed by \cite{fine2002query,gonen2013efficient}. Substituting the Monte-Carlo sampling by a Grover's search, the query complexity of the QVSP algorithm should scale as $O^*(N \sqrt{K}) = O^*(N/\sqrt{{\gamma}^{D}})$ in the worst case, instead of $O^*(N/\sqrt{{\gamma}} )$ as stated in previous result. Notwithstanding with a quantum counting subroutine~\citep{liao2024quadratic}, it scales as $O^*(\sqrt{N/\gamma^D})$. This is a negative result that suggests that the QVSP algorithm can be ineffective when applied to high-dimensional data with small margins.

\subsection{Simulation of Margin Decay}

We present in Figure~\ref{hard_data} a simulation result on an artificially designed hard dataset~\citep{mohri2018foundations,roget2022quantum} that has an exponentially small margin with data dimension.\footnote{This artificially designed dataset is  from~\cite[Chapter 8]{mohri2018foundations}, $\bm x_i'=(\underbrace{(-1)^i, \ldots,(-1)^i,(-1)^{i+1}}_{i \text { first components }}, 0, \ldots, 0)$, $y_i=(-1)^{i+1}$.}  It works as a concrete counterexample to show that the prior claims of $\Theta(\gamma)$ fail. 

\begin{figure}[h]
\centering
\includegraphics[width=1\textwidth]{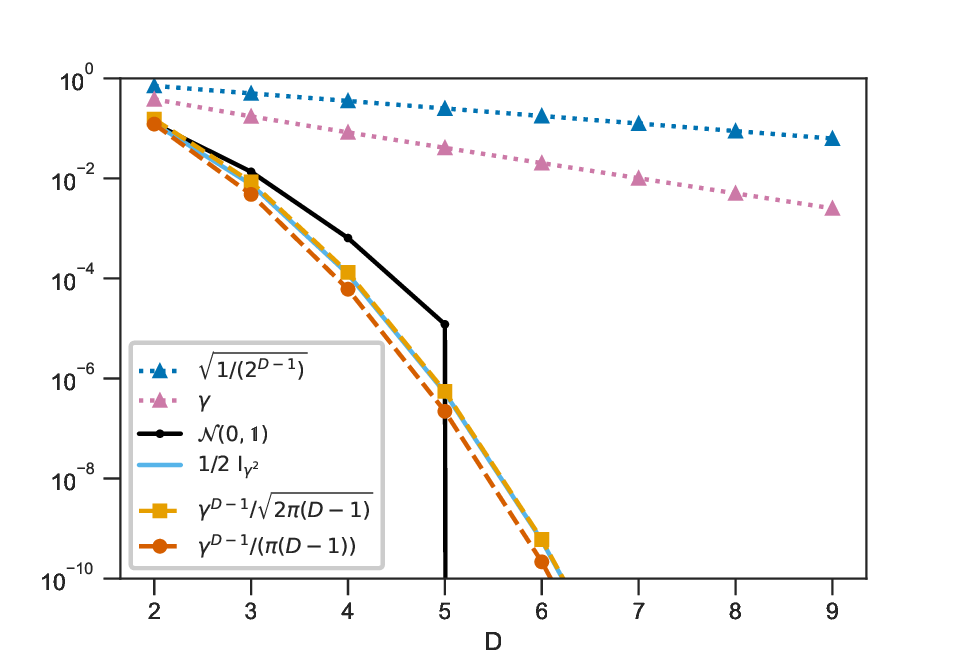}
\caption{The perfect sampling probability relationships (log scale) on a hard dataset from \cite{mohri2018foundations} (normalised norms). The black solid line denotes the perfect sampling probability; the sky blue solid line denotes the bound in  Lemma~\ref{lemma_2}; the vermilion dashed line and the orange dashed line represent the bounds from Eqn.\@\eqref{gamma_asymp} and 
Eqn.\@\eqref{d_gamma_asymp} respectively. The margin $\gamma$ (approximated by \textsc{sklearn.svm.LinearSVC}) is indicated by the reddish purple dotted line, and is upper bounded by $\sqrt{1/2^{D-1}}$ (blue dotted line)
}\label{hard_data}
\end{figure} 

As we can see from it, the probability of sampling a $\bm w \sim \mathcal{N}(0, \mathbb{1})$ that lies in the version space (black solid line) is lower bounded by the regularised incomplete beta function given in Lemma~\ref{lemma_2} (sky blue solid line). When $D\geq 2$ is a constant, $\Pr[\bm w\in \mathcal{VS}]$ is lower bounded by $\frac{\gamma^{D-1}}{\pi (D-1)}$ (vermilion dashed line) as shown in Eqn.\@\eqref{gamma_asymp}, and when $D  \rightarrow\infty$, it is asymptotically lower bounded by $
\frac{\gamma^{D-1}}{\sqrt{2\pi (D-1)}}$~(orange dashed line), as shown in Eqn.\@\eqref{d_gamma_asymp}.
It confirms our result of Theorem~\ref{theorem_3}, that the probability of sampling a perfect classifier from a standard normal distribution is lower bounded by $\Omega(\gamma^D)$ in the worst case. 

\paragraph{Discussion of the version space sampling method}

It is worth noting that the lower bound of $\Omega(\gamma^D)$ corresponds to a worst-case scenario, which could be too loose for certain datasets, where a small margin does not necessarily imply a small version space. Especially, the datasets with low-rank data matrices, which typically arise with small sample size, sparse data or small effective/intrinsic dimension. 

Particularly, one can represent a dataset by an $X\in \mathbb{R}^{N\times D}$ data matrix, where the $i^{th}$ row represents the $\bm x_iy_i$ vector, and the rank $r$ of the matrix characterises the dimension of the linear subspace spanned by the data points. When the rank is significantly smaller than the ambient feature dimension $r\ll D$, a substantially larger version space can manifest, such that the success probability $\Pr[\bm w\in \mathcal{VS}]$ may be significantly larger than the worst-case lower bound. More precisely,  we explicate this observation in the following Remark~\ref{remark4}.

\begin{remark}\label{remark4}
Let $X:=[\bm x_1 y_1,
\cdots,\bm x_N y_N]\in \mathbb{R}^{N \times D}$ be the data matrix, and let $r = \operatorname{rank}(X)$. 
Define the subspace 
\[
\mathcal S := \operatorname{span}\{\bm x_i y_i\}_{i=1}^N \subseteq \mathbb{R}^D.
\]
Since each separability constraint $y_i \langle \bm w, x_i \rangle > 0$ depends only on the projection of $\bm w$ onto $\mathcal S$, every $\bm w \in \mathbb{R}^D$ can be uniquely decomposed as
\[
\bm w = \bm w_{\parallel} + \bm w_{\perp},
\qquad
\bm w_{\parallel} \in \mathcal S,
\quad
\bm w_{\perp} \in \mathcal S^\perp.
\] 
Since  $ \bm x_i y_i \in \mathcal S$, $\forall i\in [N]$, we have $\langle \bm w_{\perp}, x_i y_i \rangle = 0$,
and therefore $\langle \bm w, x_i y_i \rangle
= \langle \bm w_{\parallel}, x_i y_i \rangle$, the separability is fully determined by $\bm w_{\parallel}$. 
Consequently, the version space admits the product decomposition
\[
\mathcal{VS}
=
\mathcal C \oplus \mathcal S^\perp,
\]
where $\mathcal C:=\{\bm u \in \mathcal S |
\langle \bm u, x_i y_i \rangle > 0, \ \forall i \}$
is a polyhedral cone in the $r$-dimensional subspace $\mathcal S$.

Therefore, the variable governing the small-margin behaviour of
$\Pr[\bm w \in \mathcal{VS}]$ is $r$, rather than the ambient dimension $D$. 
In particular, as $\gamma \to 0$, we can still use the same expression as in Theorem~\ref{theorem_3}, that
\[
\Pr[\bm w \in \mathcal{VS}]
=\frac{V(\mathcal{VS})}{V(B_1)}.
\]
However, as the volume of version space $V(\mathcal{VS})$ is restricted only by the subspace 
$\mathcal S$ with dimension $r\ll D$, while the orthogonal complement $\mathcal S^\perp$ contributes no additional constraints, it could scale much larger compared with the scenario that $r\sim D$. Thus, this indicates that the probability is likely $\Pr[\bm w \in \mathcal{VS}]= \Omega(\gamma^r)$, up to dimension-dependent constants.

Hence, following a similar analysis of the Monte-Carlo version space sampling method, as at the end of Section~\ref{theoremand proof}, the exponential dependence of $O^*(1/\gamma^D)$ should be interpreted as a  worst-case counterexample, which  could  be substantially mitigated  when the data matrix is low-rank. Particularly, we speculate that a much lower query complexity could be incurred for QVSP with low-rank datasets, scaling as $O^*(\sqrt{N/\gamma^r})$.
\end{remark}

Nevertheless, since our goal here is to demonstrate the corrected bound, we do not fully characterise the practical effectiveness of the version space sampling method.

\section{Quantum-Enhanced Perceptron Learning Algorithms}\label{Q_enh}

As shown in the last section, the use of QVSP can be severely limited in worst-case scenarios with small margins and high feature-space dimensions. In view of this, we investigate alternative quantum-enhanced approaches to address margin dependency and the practical usage of high-dimensional data. We might as well first investigate how the OQP algorithm~\citep{kapoor2016quantum} works successfully. 

Following a classical online perceptron framework, the model continuously refines its current position in the dual space after perceiving each misclassified pair. In the end, it converges to an arbitrary point $\bm w_r \in \mathcal{VS}$. However, in a traditional online learning setup, training examples are accessed incrementally in a stream fashion, which leads to a worst-case query complexity bound that scales as $O(N)$, from deterministically checking all $N$ training pairs. To provide a fair comparison with quantum query models and define analogous quantum-classical hybrid algorithms such as OQP, \cite{kapoor2016quantum} considered a \emph{weak online learning} framework.

\subsection{Quantum Weak Online Learning}
\label{weak_online}
The weak online learning approach follows the same methodology as online learning, where data are acquired incrementally from an unknown training dataset of unknown distribution. However, it differs in that it assumes each example is sampled uniformly at random from a fixed dataset $\{(\bm{x}_1, y_1), \ldots, (\bm{x}_N, y_N)\}$. Although this may resemble batch learning, we emphasise that the learner does not have access to the full dataset \emph{at once} nor have the \emph{whole knowledge} of the dataset like in a white-box model. Instead, interaction occurs through a black-box interface that returns a training example upon request. In the classical setting, this black-box is a uniform sampler: each example is returned with probability $1/N$, and repeated sampling of the same example is permitted. Consequently, to obtain a misclassified pair,  a classical uniform sampler requires repeated sampling and verifications, which induces $O(N)$ classical queries, see Eqn.\@\eqref{Boolean}. This relaxed weak online learning assumption naturally extends to a quantum weak online learning setting, where the black-box is a quantum Grover's search subroutine. Training examples are accessed in superposition, after applying $O(\sqrt{N})$ quantum queries, as in Eqn.\@\eqref{online_q}, a specific example (e.g., a misclassified one) can be retrieved (measured) with probability $O(1)$, which is the idea underpinning the OQP algorithm~\citep{kapoor2016quantum}. Below we provide an illustration flowchart of (quantum) weak online learning algorithms for the perceptron learning problem (see Figure~\ref{qweakonline}). It summarises the key idea behind the OQP algorithm.

\begin{figure}[ht!]
\centering
\tikzstyle{io} = [rectangle, rounded corners, 
minimum width=2.5cm, 
minimum height=0.6cm,
text centered, 
draw=black, 
fill=red!20,
font=\footnotesize]

\tikzstyle{process} = [rectangle,rounded corners,  
minimum width=2cm, 
minimum height=0.6cm,
text centered,
text width=2.5cm,
draw=black,
fill=orange!20,
font=\footnotesize]

\tikzstyle{box} = [rectangle, rounded corners, 
minimum width=5.5cm,
minimum height=0.6cm,
text centered,
draw=black, 
fill=blue!20,
font=\footnotesize]

\tikzstyle{arrow} = [thick,->,>=stealth]
\tikzstyle{decision} = [diamond, draw,
fill=green!10, 
text centered, 
aspect=5,
font=\footnotesize]

\begin{tikzpicture}[node distance=0.6cm and 0.8cm]
\node (start) [io] {\footnotesize Initialize model $\bm{w}_0$};
\node (box) [box, below=of start] {\footnotesize  Uniform sampler/Grover's search};
\node (check) [decision, below=1cm of box] {\footnotesize Is $(\bm{x'}_t, y'_t)$ misclassified?};
\node (update) [process, right=1cm of check] {\footnotesize Update:  $\bm{w}_{t+1}\leftarrow f(\bm{w}_t, ~y'_t \bm{x'}_t)$};
\node (stop) [ decision, below=of check] {\footnotesize Stopping condition met?};
\node (end) [io, below=of stop] {\footnotesize Output final $\bm{w}_r$};

\draw [arrow] (start) -- (box);
\draw [arrow] (box.south) -- node[left] {\small An example $(\bm{x'}_t, y'_t)$}  (check.north);

\draw [arrow] (check.east) -- node[above] {\small Yes} (update.west);
\draw [arrow] (update.south) |- (stop.east);
\draw [arrow] (check.south) -- node[left] {\small No} (stop.north);
\draw [arrow] (stop.south) -- node[left] {\small Yes} (end.north);
\draw [arrow] (stop.west)-- node[above] {\small No} ++(-1.5,0)  |- (box.west);
\end{tikzpicture}
\caption{Flowchart of (quantum) weak online learning algorithm for perceptron learning: a (Grover's search subroutine) uniform sampler is used to sample a misclassified example pair, and the model is iteratively updated until convergence ($\bm w_r\in \mathcal{VS}$).
}
\label{qweakonline}
\end{figure}

In round $t$, a uniform sampler (Grover's search) returns a misclassified training example $(\bm{x}_t', y_t')$ with success probability at least $ 1-\delta$, where $\delta > 0$. The algorithm then checks whether the current hypothesis $\bm{w}_t$ misclassifies this example. If it does, the model is updated according to its corresponding rule, as $\bm w_{t+1}\leftarrow f(\bm{w}_t,~\bm{x}_t', y_t')$, e.g., $\bm{w}_{t+1} = \bm{w}_t + y_t' \bm{x}_t'$ for online (quantum) perceptron. If the returned example pair is classified correctly, the algorithm continues to query the black-box until a misclassified pair is found. The process repeats until a stopping condition is met at round $r$, at which point the final hypothesis $\bm{w}_r$ lies within the version space $\mathcal{VS}$. 

\subsection{Hybrid Cutting-Plane Random Walk Algorithm} \label{HCPRW}
As discussed in Section~\ref{relations},  a feasibility algorithm primitive can be a proxy for an algorithm solving the perceptron learning problem. In feasibility algorithms, each misclassified pair functions as a \emph{separating hyperplane} that separates the current solution from the feasible region. The algorithms incrementally refine the convex feasible region containing the version space using the separating hyperplane. Although such a hyperplane is  typically assumed to be provided by a \emph{separation oracle}, in practice, one needs to implement such a separation oracle. 

Under a (quantum) weak online learning setup, one can implement the separation oracle via a uniform sampler or a quantum sampler, such as Grover's search. Thus, similarly, one can gain a speedup of $O^*(\sqrt{N})$ by using a quantum approach. Following this scheme, in this section, we propose a classical-quantum hybrid algorithm, HCP-RW, based on the \emph{cutting-plane random walk} (CP-RW) algorithm~\citep{bertsimas2004solving}, as the CP methods have an optimal update bounds of \( O^*(D) \)~\citep{nemirovskij1983problem,jiang2020improved} among all feasibility algorithms.
Below, we first review the geometric analysis of the classical CP-RW algorithm.

Proposed by \cite{bertsimas2004solving}, the CP-RW algorithm provides a method to effectively reduce the volume in $\mathbb{R}^D$. In each round, the algorithm preserves an approximate centroid $\bm z_t$ of the current space, which we use as a perceptron model similarly to the work of \cite{louche2015cutting}. Initialized being the centre of $\mathcal{P}_0=B_1$, the perceptron updates its position by $\bm z_{t+1} \gets \frac{1}{M}\sum_{j=1}^{M} \bm w'_{j,t}$~(Algorithm~\ref{Algo_HCPRW}), where $\bm w'_{j,t}$ are uniformly sampled vectors over the remaining space $\{\bm w\mid(\bm w -\bm z_t)^T \bm x_t'y_t'>0\}$ cut by a misclassified example $(\bm x'_t,y'_t)$. When the sample size $N$ is small, sampling $2M=O (\log^2(N))\leq O(D)$ points suffices that the volume cut-off in each round is at least $1/3$ with a high probability~\citep{bertsimas2004solving}. From Eqn.\@\eqref{d_gamma_asymp}, one has $V(\mathcal{VS})/V(\mathcal{P}_0)\geq\frac{\gamma^{D-1}}{\pi(D-1)}$, so letting $\frac{V(\mathcal{P}_r)}{V(\mathcal{P}_0)}\leq(\frac{2}{3})^r\leq\frac{\gamma^{D-1}}{\pi(D-1)}\leq \frac{V(\mathcal{VS})}{V(\mathcal{P}_0)} $ ensures $V(\mathcal{P}_r)\leq V(\mathcal{VS})$ and $\bm w_r\in \mathcal{VS}$. A direct calculation yields $r \geq (D-1)\log_{\frac{3}{2}}({1}/{\gamma})+\log_{\frac{3}{2}} ({\pi(D-1)}  )$, hence the stop criterion can be $r= \lceil D\log_{3/2} ({D}/{\gamma})  \rceil$. 

To ensure fast sampling, we adapt the \emph{Hit-and-Run} algorithm proposed by \cite{smith1984efficient} with a uniform sampling scheme. Starting from a \emph{warm start} distribution\footnote{A warm start is a starting distribution that is almost close to the stationary distribution.}, for all $\bm x\in \mathcal{P}_t$, it works as follows:
\begin{equation*}
    \begin{aligned}
& \textbf{Hit-and-Run}\\
& \text { (1) Choose a line } \ell \text { through the current point } \bm x \text { uniformly at random. } \\
& \text { (2) Move to a point }\bm y \text { chosen \textbf{uniformly} from } \mathcal{P}_t \cap \ell \text {. }
\end{aligned}
\end{equation*}

It is known that given a warm start, it has the best-known mixing bound such that it generates a stationary distribution inside of a well-rounded convex body $\mathcal{K}$ in $O^*(D^3)$ steps ~\citep{lovasz1999hit}. Particularly, this stationary distribution is uniform since it follows a uniform sampling scheme. In the $t^{th}$ round, starting from a current point $\bm w\in \mathcal{P}_t$ and giving a \emph{membership oracle} $\mathcal{O}_{\mathcal{P}_t}$ defined on the current convex body $\mathcal{P}_t$, each step of walk costs $O(\log (1/\epsilon) )$ calls to $\mathcal{O}_{\mathcal{P}_t}$ leveraging binary search. However, to implement a membership oracle $\mathcal{O}_{\mathcal{P}_t}$ in practice,  one needs to query $t$ example pairs, where
\begin{equation}
\label{single_mem_Boolean}
\begin{aligned}
&\mathcal{O}_{\mathcal{P}_t}(\bm w)= f_{\bm w} (\bm x'_1 y'_1,\bm z_{0})\lor f_{\bm w} (\bm x'_2y'_2,\bm z_{1}) \lor \cdots \lor f_{\bm w}(\bm x'_t y'_t,\bm z_{t-1}),
\end{aligned}\end{equation} 
taking a similar form as Eqn.\@\eqref{version_q}. 
Here the Boolean functions are defined similarly as Eqn.\@\eqref{Boolean}, where $f_{\bm w} (\bm x'_jy'_j,\bm z_{j-1}) = 0 $ iff $(\bm w-\bm z_{j-1})^T \bm x'_j y'_j>0$, $\forall j \in[t]$. Herewith,  $\mathcal{O}_{\mathcal{P}_t}(\bm w)= 0$ iff $\bm w\in \mathcal{P}_t$. Consequently, it takes $O(t \log(1/\epsilon) )$ classical quires  to implement the membership oracle $\mathcal{O}_{\mathcal{P}_t}$ classically. Overall, this classical cutting-plane algorithm with the Hit-and-Run walk has a classical  query complexity of
$ O^* (D\log{({1}/{\gamma})}\cdot (N + \log^2N\cdot D^3 \cdot r )= O^* (D\log{({1}/{\gamma})}\cdot (N + D^5\log{({1}/{\gamma})} )$.

Following a CP-RW framework, we present the HCP-RW algorithm as shown in Algorithm~\ref{Algo_HCPRW}. It adopts Grover's search (shown as $Qsearch$ in Algorithm~\ref{Algo_HCPRW}) under the quantum weak online learning scheme and outputs a vector $\bm w_r\in \mathcal{VS}$ with a probability of failure at most $\epsilon$.  

\begin{algorithm}[ht!]
\caption{HCP-RW}\label{Algo_HCPRW}
\begin{algorithmic}
\State $\textbf{Input}:\{(\bm x_i,y_i)_{i\in[N]}\} $
\State $\textbf{Initialize}:$
\State $ W :=\{\bm w_1,\ldots,\bm w_{M}\} \in B_1$,  $\bm z := \bm 0$, $S(\bm x):=\frac{\bm x-\bm z}{\sqrt{\mathbb{E}[(\bm x-\bm z)(\bm x-\bm z)^T]}} $
\For{$i \in \{1,\ldots, \lceil D\log_{3/2}{ (\frac{D}{\gamma}  )}   \rceil  \}$}
\For{ $j \in \{1, \ldots, \lceil\log _{3 / 4} (\frac{\epsilon}{D\log_{3/2}{ (D/\gamma  )} }  )  \rceil  \}$}
\State{$(\bm x',y')\gets Qsearch (\bm z,\{(\bm x_i,y_i)_{i\in[N]}\}  )$} \Comment{with success probability $\geq \frac{1}{4}$.}
\If{$\bm z^T \bm x'y'\leq 0$}
\State{$  \textbf{delete } \bm w_i \textbf{ from } W \textbf{ if } f_{\bm w}(\bm x'y',\bm z)=1, \forall i \in [M]$.} 
\State $\textbf{calculate the affine transform function } S $
\State $W \gets \{\bm w_1,\ldots,\bm w_{2M}\} \gets \textit{Hit-and-Run }(W,S,\mathcal{Q}_{\mathcal{P}})$ 
\State $ W'\gets \{\bm w'_1,\ldots,\bm w'_{M}\} \gets\text{random sample } M \text{ points from }W $
\State $\bm z \gets \frac{1}{M}\sum_{j=1}^{M} \bm w'_j$
\State $\textbf{discard } W' \textbf{ from } W$ \Comment{ avoid
correlations with future samples.}
\EndIf
\EndFor
\EndFor
\end{algorithmic}
\end{algorithm}

Instead of using the classical membership oracle  $\mathcal{O}_{\mathcal{P}_t}$, the HCP-RW algorithm implements a quantum membership oracle (shown as $\mathcal{Q}_{\mathcal{P}}$ in Algorithm~\ref{Algo_HCPRW}) by leveraging a quantum counting algorithm as of \cite{liao2024quadratic}. Specifically,  $\mathcal{Q}_{\mathcal{P}_t}$ can be implemented  with a small probability of failure by using $O(\sqrt{t})$ quantum queries defined  similarly with Eqn.\@\eqref{online_q}
\begin{equation}\label{quan_mem}
    F\ket{w}|x_j'y_j'\rangle\ket{z_{j-1}} \ket{0}=\ket{w}|x_j'y_j'\rangle\ket{z_{j_1}}\ket{0\oplus {f_{\bm w} (\bm x_j',y_j',\bm z_{j-1})}}.
\end{equation}

Taking the quantum speedups into account, the overall quantum query complexity is $O^*(D\log{({1}/{\gamma})}\cdot(\sqrt{N} +\log^2 N \cdot D^3\cdot\sqrt{t})
$, substituting $t\leq r=O(D\log(1/\gamma))$, it is asymptotically bounded by $O^*(D\log{({1}/{\gamma})}\cdot(\sqrt{N} +\log^2 N \cdot D^{3.5}\cdot\sqrt{\log (1/\gamma)}
= O^*(D\log{({1}/{\gamma})}\cdot(\sqrt{N} +D^{4.5}\cdot\sqrt{\log (1/\gamma)}$. Additionally, in both HCP-RW and CP-RW algorithms, affine transforms are required to ensure the convex bodies are well-rounded; it incurs $O(D^2)$ arithmetic operations, so their arithmetic operations are bounded by  $O^*(D\log{({1}/{\gamma})}\cdot(\log^2 N \cdot D^3\cdot D^2)=O^*(\log^2 N \cdot\log(1/\gamma)\cdot D^6)=O^*(D^7 \log(1/\gamma))$.

From another perspective, it can also be seen clearly how quantum algorithms/subroutines can speed up solving a feasibility problem. In a similar vein, there are flourishing advancements in quantum linear programs and semidefinite programs algorithms. For example, \cite{kerenidis2020quantum,apers2026quantum} developed quantum algorithms to speed up the interior point method. In the former, the authors emphasise an efficient quantum algorithm for constructing block encodings to solve linear systems. In the latter, the authors deliberate on using Grover's search to efficiently approximate the Newton linear system matrices, and based on which, \cite{li2024quantum} proposed a quantum algorithm to approximate the Löwner–John  ellipsoid. 

\subsection {Fully-Quantum Cutting-Plane Quantum Walk Algorithm}
\label{FQCP_QW}

Realising that for the HCP-RW algorithm, a query complexity of $O^*(D^{5.5})$ is still fairly high when $D$ is large. However, it has been proved that the cutting plane methods achieve the optimal number of calls of the separation oracle $\Omega^*(D) $ for solving the feasibility problem. Remarkably, it has also been shown by \cite{van2020convex} that quantum algorithms cannot improve this lower bound under a black-box interface. Yet we can see that using a random walk sampler incurs a high computational cost of $O^*(D^3)$. 

In this section, we aim to improve the computational complexity by applying alternative quantum enhancements in the sampling process. Specifically, we propose a fully-quantum algorithm, the \emph{cutting-plane quantum walk} (QCP-QW) algorithm in Section~\ref{QCP_QW}, adopting the quantum weak online learning framework as HCP-RW, but differs in that both the training dataset and the \emph{model itself} in the dual (parameter) space are represented in quantum superposition.
 
On one hand, it preserves the strengths of the CP method, achieving a margin dependence of \( O^*(D \log (1/\gamma)) \). On the other hand, its algorithmic structure is inherently quantum: both the dataset and the perceptron model are encoded as quantum states, and  quantum operations are performed coherently.  Specifically, the algorithm incorporates the \emph{hit-and-run quantum walk} framework~\citep{chakrabarti2023quantum,li2022quantum}, which leverages the quantum walk search algorithm introduced by \cite{magniez2007search} along with other quantum subroutines, such as quantum mean estimation and affine transformation estimation. As a result, the quantum perceptron model itself is updated continuously without destructive measurements. Compared with the HCP-RW algorithm, the sampling process is done quantumly instead of classically, which is the fundamental reason for gaining an extra speedup of $O^*(D^{1.5})$.

\subsubsection{Szegedy's Quantum Walk}\label{Szegedy}

In the Appendix~\ref{Markov_chian}, we introduce some preliminaries of  classical and quantum Markov chains. 

The Szegedy's quantum walk is defined based on an  \emph{ergodic} and \emph{time-reversible} Markov chain $P$, it works on the Hilbert space \( \mathcal{H} = \mathbb{C}^N \otimes \mathbb{C}^N \) with basis states \( \ket{\bm x}\ket{\bm y} \) for \( \bm x,\bm y \in \Omega \).  The Szegedy's quantum walk operator $W(P): = R_{\mathcal{B}} R_{\mathcal{A}}$, where  $R_{\mathcal{I}}=2\operatorname\Pi_{\mathcal{I}}-\mathbb{1}$ is the reflection about the $\mathcal{I}$ subspace, with subspaces $\mathcal{A}=\operatorname{span}\{\ket{\bm x}\ket{\bm 0}: \bm x\in \operatorname\Omega\}$ and $\mathcal{B}=\operatorname{span}\{\ket{\bm 0}\ket{\bm x}: \bm x\in \operatorname\Omega\}$. It was shown by \cite{szegedy2004quantum,magniez2007search} that for an \emph{irreducible} and time-reversible Markov chain \( P \), the quantum walk operator \( W(P) \) has a unique eigenvalue-1 eigenvector
$$
\ket{\pi} = \ket{\pi}' \ket{\bm 0} \in \mathcal{H},
$$
in the subspace \( \mathcal{A} \cap \mathcal{B} \), where $\ket{\pi}'$ is a quantum sample\footnote{A quantum sample corresponding to a classical probability distribution \( f \) can is defined as
$\ket{f} := \sum_{\bm x \in \Omega} \sqrt{f(\bm x)}\ket{\bm x}$, with $\sum_{\bm x \in \Omega} |f(\bm x)|^2 = 1$.} of the stationary distribution \( \bm{\pi} \) of \( P \), given by
$$
\ket{\pi}' = \sum_{\bm x \in \Omega} \sqrt{\pi_{\bm x}} \ket{\bm x}.
$$
The eigenvalues of \( W(P) \) with non-zero imaginary parts are exactly \( e^{\pm 2 \pi i \theta_j} \), where
$$
\cos \theta_j = |\lambda_j|,
$$
with \( j \in [M] \), \( M \leq N-1 \), on the subspace \( \mathcal{A} + \mathcal{B} \). These \( 2M \) eigenvalues satisfy \( \theta_j \in (0, \pi/2) \) and \( 0 < |\lambda_j| < 1 \) for \( P \). In the following, we consider quantum walks acting on the subspace \( \mathcal{A} + \mathcal{B} \), since their action on the orthogonal complement is trivial (eigenvalues \( \pm 1 \)). 

A quantum-walk-update for \( P_t \) is given by a unitary operator \( U(P_t) \) that satisfies, for all \( \bm x \in \Omega \) and fixed \( \bm 0 \in \Omega \),
\begin{equation}\label{quantum_update}
U(P_t) \ket{\bm x}\ket{\bm 0} = \ket{\bm x} \ket{\bm p_x} = \ket{\bm x} \sum_{\bm y \in \Omega} \sqrt{p_{xy}} \ket{\bm y}.
\end{equation}
The quantum walk operator for the \( t^{th} \) Markov chain $W_t $ can be implemented as
\begin{equation}\label{quantum_walk}
    W_t := W(P_t) = U(P_t)^\dagger S U(P_t) R_{\mathcal{A}} U(P_t)^\dagger S U(P_t) R_{\mathcal{A}},
\end{equation}
where \( S \) denotes the swap operator.

\subsubsection{Algorithm Construction and Analysis}\label{QCP_QW}

Below, we provide a general description of our QCP-QW algorithm with pseudocode (Algorithm~\ref{Algo_QCPQW}) and an illustration picture (Figure~\ref{AF}).

\begin{itemize}
\item \textbf{High-level:}  
The QCP-QW algorithm follows the quantum weak online learning framework as illustrated in Figure~\ref{qweakonline}, similar to the HCP-RW algorithm. Specifically, the inner loop employs Grover's search ($Qsearch$ in Algorithm~\ref{Algo_QCPQW})  to ensure the probability of finding a separating hyperplane is at least $1/4$. Relying on the CP method as detailed in Section~\ref{HCPRW}, the outer loop performs \( O^*(D \log(1/\gamma)) \) updates to guarantee the final solution \( \ket{\pi_r} \) lies inside the version space with small error. 
    
    \item \textbf{Middle-level:}  
    Apply the non-destructive quantum mean estimation procedure \( U^{\operatorname{mean}}_t \) on the state \( \ket{\tilde{\pi}_t} \) to acquire the multivariate mean \( \bm{z}_{t+1} \) of \( \tilde{\bm{\pi}}_{t+1} \). Similarly, apply the non-destructive affine estimation procedure \( U^{\operatorname{aff}}_t \) on \( M = O(D) \) copies of \( \ket{\tilde{\pi}_t} \) to obtain the affine transformation matrix \( S_{t+1} \) of \( \tilde{\bm{\pi}}_{t+1} \).
    
    \item \textbf{Low-level:}  
    Apply the affine transformation function \( g_t \) to the states \( \ket{\tilde{\pi}_t} \) to ensure the convex body \( \mathcal{P}_t \) is well-rounded.\footnote{A convex body is well-rounded if it contains a ball of radius \( r \) and is contained in a ball of radius \( R \), such that \( R/r = O(\sqrt{D}) \).} Then run the quantum-walk-evolution unitary \( \tilde{U}_{t,m} \) to generate the distribution \( \ket{\tilde{\pi}_{t+1}} \) via
    $$
        \ket{g_t \tilde{\pi}_{t+1}} = \tilde{U}_{t,m} \ket{g_t \tilde{\pi}_t}.
    $$
\end{itemize}

\begin{algorithm}[ht]
\caption{QCP-QW} \label{Algo_QCPQW}
\begin{algorithmic}[1]
    \State \textbf{Input:} \(\{(\bm{x}_i, y_i)\}_{i \in [N]}\)
    \State \textbf{Initialize:}
    \State \quad \( \ket{\pi_0}\ket{0}^{\otimes ac} := \sum_{w \in B_{1,\epsilon}} \sqrt{\frac{1}{|B_{1,\epsilon}|}} \ket{w} \ket{0}^{\otimes ac} \)
    \State \quad \( \bm{z}_0 := \bm{0}, \quad S_0 = \mathbb{1}, \quad g_0(\bm{w}) := S_0 (\bm{w} - \bm{z}_0) \)
    \For{\( i = 1, \ldots, \lceil D \log_{3/2}(D/\gamma) \rceil \)}
        \For{\( j = 1, \ldots, \lceil \log_{3/4}(\frac{\epsilon}{D \log_{3/2}(D/\gamma)}) \rceil \)}
            \State \((\bm{x}', y') \gets Qsearch(\bm{z}_0, \{(\bm{x}_i, y_i)\}_{i \in [N]})\) \Comment{success probability \(\geq 1/4\)}
            \If{\( \bm{z}^T \bm{x}' y' \leq 0 \)}
                \State (1) Update \(\ket{g_t \tilde{\pi}_{t+1}} = \tilde{U}_{t,m} \ket{g_t \tilde{\pi}_t}\) and then invert to obtain \(\ket{\tilde{\pi}_{t+1}}\).
                \State (2) Apply \( U^{\operatorname{mean}}_{t+1} \) on the state \( \ket{\tilde{\pi}_t} \) to acquire the 
                \State ~ multivariate mean \( \bm{z}_{t+1} \) of \( \tilde{\bm{\pi}}_{t+1} \) non-destructively.
                \State (3) Apply \( U^{\operatorname{aff}}_{t+1} \) on \( M = O(D) \) copies of \( \ket{\tilde{\pi}_t} \) to acquire the 
                \State ~ covariance matrix \( S_{t+1} \) of \( \tilde{\bm{\pi}}_{t+1} \) non-destructively.
            \EndIf
        \EndFor
    \EndFor
\end{algorithmic}
\end{algorithm}

\begin{figure}[ht!]
\centering
\tikzset{meter/.append style={scale=0.6, draw, inner sep=10, rectangle, font=\vphantom{A}, minimum width=30, line width=.8,
 path picture={\draw[black] ([shift={(.1,.3)}]path picture bounding box.south west) to[bend left=50] ([shift={(-.1,.3)}]path picture bounding box.south east);\draw[black,-latex] ([shift={(0,.1)}]path picture bounding box.south) -- ([shift={(.3,-.1)}]path picture bounding box.north);}}}

 \scalebox{0.8}{
\begin{tikzpicture}
\def\qubitsep{1.0}
\def\colsep{1}
\foreach \i in {0,1} {
\draw [thick]  (0, -\i*1.5*\qubitsep) -- (10.3*\colsep, -\i*1.5*\qubitsep);}

\node at (-0.2*\colsep, -0.75*\qubitsep) {\vdots};
\node at (-0.2*\colsep, -3.7*\qubitsep) {\vdots};
\node  at (-0.2*\colsep, -7.3*\qubitsep) {\vdots};
\node  at (-0.2*\colsep, -6.9*\qubitsep) {\vdots};

\draw [thick] (0, -3*\qubitsep) -- (10.3*\colsep, -3*\qubitsep);
\draw [thick] (0, -4.5*\qubitsep) -- (10.3*\colsep, -4.5*\qubitsep);
\draw [thick] (0, -5.5*\qubitsep) -- (10.3*\colsep, -5.5*\qubitsep);
\draw [thick] (0, -6.5*\qubitsep) -- (10.3*\colsep, -6.5*\qubitsep);
\draw [thick] (0, -8*\qubitsep) -- (10.3*\colsep, -8*\qubitsep);

\draw [thick] (10.9, -5.5*\qubitsep) -- (13.5*\colsep, -5.5*\qubitsep);
\draw [thick] (10.9, -6.5*\qubitsep) -- (13.5*\colsep, -6.5*\qubitsep);
\draw [thick] (10.9, -8*\qubitsep) -- (13.5*\colsep, -8*\qubitsep);

\node[left] at (0, -4.5*\qubitsep) {$\ket{0}$};

\node[left] at (0, -3*\qubitsep) {$\ket{0}$};
\node[left] at (0, -5.5*\qubitsep) {$\ket{0}$};
\node[left] at (0, -6.5*\qubitsep) {$\ket{0}$};
\node[left] at (0, -8*\qubitsep) {$\ket{0}$};

\draw[fill=white, thick] (0.5*\colsep, 0.4*\qubitsep) rectangle (1.7*\colsep, -5*\qubitsep);
 \node at (1.1*\colsep, -2.5*\qubitsep) {$U_{t+1}^{mean}$};
\draw [dashed] (2*\colsep,0.3*\qubitsep)-- (2*\colsep, -2*\qubitsep);
\node[right] at (1.7*\colsep, 0.6*\qubitsep) {$\ket{\tilde{\pi}_{t}}^{\otimes M_1}$};

\draw[fill=black] (2.5*\colsep, -3*\qubitsep) circle (2pt);
\draw[fill=black] (2.5*\colsep, -4.5*\qubitsep) circle (2pt);
\node  at (2.5*\colsep, -3.7*\qubitsep) {\vdots};
\draw [thick] (2.5*\colsep, -3.1*\qubitsep)-- (2.5*\colsep, -3.5*\qubitsep);
\draw [thick] (2.5*\colsep, -5.5*\qubitsep)-- (2.5*\colsep, -4.2*\qubitsep);
\draw[fill=white,thick] (1.8*\colsep, -5.2*\qubitsep) rectangle (3.2*\colsep, -5.8*\qubitsep); \node at (2.5*\colsep, -5.5*\qubitsep) {median };

\draw[fill=white,thick] (3.3*\colsep, 0.4*\qubitsep) rectangle (4.3*\colsep, -5*\qubitsep);
\node at (3.8*\colsep, -2.5*\qubitsep) {$\tilde{U}_{t,m}$};

\draw [dashed] (4.6*\colsep,0.3*\qubitsep)-- (4.6*\colsep, -2*\qubitsep);
\node[right] at (4.3*\colsep, 0.6*\qubitsep) {$\ket{\tilde{\pi}_{t+1}}^{\otimes M_1}$};

\draw[fill=white, thick] (5.9*\colsep, 0.4*\qubitsep) rectangle (7.1*\colsep, -5*\qubitsep);
 \node at (6.5*\colsep, -2.5*\qubitsep) {$U_{t+2}^{mean}$};
 \draw [dashed] (7.4*\colsep,0.3*\qubitsep)-- (7.4*\colsep, -2*\qubitsep);
\node[right] at (7.1*\colsep, 0.6*\qubitsep) {$\ket{\tilde{\pi}_{t+1}}^{\otimes M_1}$};
 
 \draw[fill=black] (7.9*\colsep, -3*\qubitsep) circle (2pt);
\draw[fill=black] (7.9*\colsep, -4.5*\qubitsep) circle (2pt);
\node  at (7.9*\colsep, -3.7*\qubitsep) {\vdots};
\draw [thick] (7.9*\colsep, -3*\qubitsep)-- (7.9*\colsep, -3.5*\qubitsep);
\draw [thick] (7.9*\colsep, -6.5*\qubitsep)-- (7.9*\colsep, -4.2*\qubitsep);
\draw[fill=white,thick] (7.2*\colsep, -6.2*\qubitsep) rectangle (8.6*\colsep, -6.8*\qubitsep); \node at (7.9*\colsep, -6.5*\qubitsep) {median};

\draw[fill=white,thick] (8.7*\colsep, 0.4*\qubitsep) rectangle (9.9*\colsep, -5*\qubitsep);
\node at (9.3*\colsep, -2.5*\qubitsep) {$\tilde{U}_{t+1,m}$};
 \draw [dashed] (10.2*\colsep,0.3*\qubitsep)-- (10.2*\colsep, -2*\qubitsep);
\node[right] at (9.9*\colsep, 0.6*\qubitsep) {$\ket{\tilde{\pi}_{t+2}}^{\otimes M_1}$};

\node  at (11.9*\colsep, -7.3*\qubitsep) {\vdots};
\node  at (11.9*\colsep, -6.9*\qubitsep) {\vdots};

\draw[fill=white,thick] (11.2*\colsep, -7.7*\qubitsep) rectangle (12.6*\colsep, -8.3*\qubitsep); \node at (11.9*\colsep, -8*\qubitsep) {median};

\node[meter] (meter) at (13.4*\colsep, -5.5*\qubitsep) {};
\node[meter] (meter) at (13.4*\colsep, -6.5*\qubitsep) {};
\node[meter] (meter) at (13.4*\colsep, -8*\qubitsep) {};

\node[right]  at (13.8*\colsep, -5.5*\qubitsep) {$\bm z_{t+1}$};
\node[right]  at (13.8*\colsep, -6.5*\qubitsep) {$\bm z_{t+2}$};
\node[right]  at (13.8*\colsep, -8*\qubitsep) {$\bm z_{r+1}$};
\draw [decorate,decoration={brace,amplitude=4pt,mirror},xshift=0.5cm,yshift=0pt,thick]
    (-0.9,0.3) -- (-0.9,-1.8) node [midway,left,xshift=-0.1cm] {$\ket{\tilde{\pi}_t}^{\otimes M_1}$};    
      
\node[right] at (10.3*\colsep, 0*\qubitsep) {$\cdots$};
\node[right] at (10.3*\colsep, -1.5*\qubitsep) {$\cdots$};
\node[right] at (10.3*\colsep, -3*\qubitsep) {$\cdots$};
\node[right] at (10.3*\colsep, -4.5*\qubitsep) {$\cdots$};
\node[right] at (10.3*\colsep, -5.5*\qubitsep) {$\cdots$};
\node[right] at (10.3*\colsep, -6.5*\qubitsep) {$\cdots$};
\node[right] at (10.3*\colsep, -8*\qubitsep) {$\cdots$};

\node at (-0.2*\colsep, -10.75*\qubitsep) {\vdots};
\node at (-0.2*\colsep, -13.7*\qubitsep) {\vdots};
\node  at (-0.2*\colsep, -17.3*\qubitsep) {\vdots};
\node  at (-0.2*\colsep, -16.9*\qubitsep) {\vdots};

\draw [thick] (0, -10*\qubitsep) -- (10.3*\colsep, -10*\qubitsep);
\draw [thick] (0, -11.5*\qubitsep) -- (10.3*\colsep, -11.5*\qubitsep);
\draw [thick] (0, -13*\qubitsep) -- (10.3*\colsep, -13*\qubitsep);
\draw [thick] (0, -14.5*\qubitsep) -- (10.3*\colsep, -14.5*\qubitsep);
\draw [thick] (0, -15.5*\qubitsep) -- (10.3*\colsep, -15.5*\qubitsep);
\draw [thick] (0, -16.5*\qubitsep) -- (10.3*\colsep, -16.5*\qubitsep);
\draw [thick] (0, -18*\qubitsep) -- (10.3*\colsep, -18*\qubitsep);

\draw [thick] (10.9, -15.5*\qubitsep) -- (13.5*\colsep, -15.5*\qubitsep);
\draw [thick] (10.9, -16.5*\qubitsep) -- (13.5*\colsep, -16.5*\qubitsep);
\draw [thick] (10.9, -18*\qubitsep) -- (13.5*\colsep, -18*\qubitsep);

\node[left] at (0, -14.5*\qubitsep) {$\ket{0}$};
\node[left] at (0, -13*\qubitsep) {$\ket{0}$};
\node[left] at (0, -15.5*\qubitsep) {$\ket{0}$};
\node[left] at (0, -16.5*\qubitsep) {$\ket{0}$};
\node[left] at (0, -18*\qubitsep) {$\ket{0}$};

\draw[fill=white, thick] (0.5*\colsep, -9.6*\qubitsep) rectangle (1.7*\colsep, -15*\qubitsep);
 \node at (1.1*\colsep, -12.5*\qubitsep) {\small $U_{t+1}^{aff}$};
 
\draw [dashed] (2*\colsep,-9.7*\qubitsep)-- (2*\colsep, -12*\qubitsep);
\node[right] at (1.7*\colsep, -9.4*\qubitsep) {$\ket{\tilde{\pi}_{t}}^{\otimes M_2}$};

\draw[fill=black] (2.5*\colsep, -13*\qubitsep) circle (2pt);
\draw[fill=black] (2.5*\colsep, -14.5*\qubitsep) circle (2pt);
\node  at (2.5*\colsep, -13.7*\qubitsep) {\vdots};
\draw [thick] (2.5*\colsep, -13.1*\qubitsep)-- (2.5*\colsep, -13.5*\qubitsep);
\draw [thick] (2.5*\colsep, -15.5*\qubitsep)-- (2.5*\colsep, -14.2*\qubitsep);
\draw[fill=white,thick] (1.8*\colsep, -15.2*\qubitsep) rectangle (3.2*\colsep, -15.8*\qubitsep); \node at (2.5*\colsep, -15.5*\qubitsep) {median };

\draw[fill=white,thick] (3.3*\colsep, -9.6*\qubitsep) rectangle (4.3*\colsep, -15*\qubitsep);
\node at (3.8*\colsep, -12.5*\qubitsep) {$\tilde{U}_{t,m}$};

\draw [dashed] (4.6*\colsep,-9.7*\qubitsep)-- (4.6*\colsep, -12*\qubitsep);
\node[right] at (4.3*\colsep, -9.4*\qubitsep) {$\ket{\tilde{\pi}_{t+1}}^{\otimes M_2}$};

\draw[fill=white, thick] (5.9*\colsep, -9.6*\qubitsep) rectangle (7.1*\colsep, -15*\qubitsep);
 \node at (6.5*\colsep, -12.5*\qubitsep) {$U_{t+2}^{aff}$};
 \draw [dashed] (7.4*\colsep,-9.7*\qubitsep)-- (7.4*\colsep, -12*\qubitsep);
\node[right] at (7.1*\colsep, -9.4*\qubitsep) {$\ket{\tilde{\pi}_{t+1}}^{\otimes M_2}$};
 
 \draw[fill=black] (7.9*\colsep, -13*\qubitsep) circle (2pt);
\draw[fill=black] (7.9*\colsep, -14.5*\qubitsep) circle (2pt);
\node  at (7.9*\colsep, -13.7*\qubitsep) {\vdots};
\draw [thick] (7.9*\colsep, -13*\qubitsep)-- (7.9*\colsep, -13.5*\qubitsep);
\draw [thick] (7.9*\colsep, -16.5*\qubitsep)-- (7.9*\colsep, -14.2*\qubitsep);
\draw[fill=white,thick] (7.2*\colsep, -16.2*\qubitsep) rectangle (8.6*\colsep, -16.8*\qubitsep); \node at (7.9*\colsep, -16.5*\qubitsep) {median};

\draw[fill=white,thick] (8.7*\colsep, -9.6*\qubitsep) rectangle (9.9*\colsep, -15*\qubitsep);
\node at (9.3*\colsep, -12.5*\qubitsep) {$\tilde{U}_{ t+1,m}$};
 \draw [dashed] (10.2*\colsep,-9.7*\qubitsep)-- (10.2*\colsep, -12*\qubitsep);
\node[right] at (9.9*\colsep, -9.4*\qubitsep) {$\ket{\tilde{\pi}_{t+2}}^{\otimes M_2}$};

\node  at (11.9*\colsep, -17.3*\qubitsep) {\vdots};
\node  at (11.9*\colsep, -16.9*\qubitsep) {\vdots};

\draw[fill=white,thick] (11.2*\colsep, -17.7*\qubitsep) rectangle (12.6*\colsep, -18.3*\qubitsep); \node at (11.9*\colsep, -18*\qubitsep) {median};

\node[meter] (meter) at (13.4*\colsep, -15.5*\qubitsep) {};
\node[meter] (meter) at (13.4*\colsep, -16.5*\qubitsep) {};
\node[meter] (meter) at (13.4*\colsep, -18*\qubitsep) {};

\node[right]  at (13.8*\colsep, -15.5*\qubitsep) {$S_{t+1}$};
\node[right]  at (13.8*\colsep, -16.5*\qubitsep) {$ S_{t+2}$};
\node[right]  at (13.8*\colsep, -18*\qubitsep) {$S_{r+1}$};

\draw [decorate,decoration={brace,amplitude=4pt,mirror},xshift=0.5cm,yshift=0pt,thick]
    (-0.9,-9.7) -- (-0.9,-11.8) node [midway,left,xshift=-0.1cm] {$\ket{\tilde{\pi}_t}^{\otimes M_2}$};    
      
\node[right] at (10.3*\colsep, -10*\qubitsep) {$\cdots$};
\node[right] at (10.3*\colsep, -11.5*\qubitsep) {$\cdots$};
\node[right] at (10.3*\colsep, -13*\qubitsep) {$\cdots$};
\node[right] at (10.3*\colsep, -14.5*\qubitsep) {$\cdots$};
\node[right] at (10.3*\colsep, -15.5*\qubitsep) {$\cdots$};
\node[right] at (10.3*\colsep, -16.5*\qubitsep) {$\cdots$};
\node[right] at (10.3*\colsep, -18*\qubitsep) {$\cdots$};

\end{tikzpicture}
}
\caption{An illustration of the QCP-QW algorithm (assuming well-rounded). The upper one acts on $M_1$ copies of samples, and $U^{mean}$ is the non-destructive mean estimation circuit. The lower one acts on $M_2$ copies, with $U^{aff}$ denoting the circuit for non-destructively estimating the covariance matrix. Here the $\tilde{U}_{t,m}$ is the quantum-walk-evolution unitary that drives state from $\ket{\tilde{\pi}_t}$ to $\ket{\tilde{\pi}_{t+1}}$, and $\mathrm{median}$ denotes taking the median over multiple copies' results. }
    \label{AF}
\end{figure}

Since the classical Hit-and-Run~\citep{smith1984efficient} is defined on continuous spaces, we adapt the analysis tool of the continuous-space hit-and-run quantum walk framework, similar to \cite{chakrabarti2023quantum} and \cite{li2022quantum}. For a uniform distribution \( \bm{\pi}_t \) over \( \mathcal{P}_t \), its quantum sample is
$$
    \ket{\pi_t}' = \int_{\mathcal{P}_t} \sqrt{\frac{1}{V(\mathcal{P}_t)}} \ket{\bm x} \, \mathrm{d}\bm x,
$$
where
$$
    \ket{\pi_t} = \int_{\mathcal{P}_t} \sqrt{\pi_t(x)} \ket{\bm x} \otimes \ket{\psi_x} \, \mathrm{d}\bm x,
    \quad \text{with} \quad
    \ket{\psi_x} := \int_{\mathcal{P}_t} \sqrt{p_{xy}} \ket{\bm y} \,\mathrm{d}\bm y,
$$
defined on the support \( \mathcal{P}_t \subseteq \mathbb{R}^D \).

For digital implementation, we adopt the \(\epsilon\)-discretisation suggested by \cite{chakrabarti2023quantum}, which provides a detailed analysis of discretisation errors. Specifically, each coordinate of any vector inside the convex body \( \mathcal{K} \subseteq \mathbb{R}^D \) is represented using \( O(\log(1/\epsilon)) \) bits. This finite set of vectors \( \mathcal{K}_\epsilon \) is called an \(\epsilon\)-net of \( \mathcal{K} \). Assume access to an initial quantum state
$$
    \ket{\pi_0} = \ket{\pi_0}' \ket{\bm 0}
$$
over \( \Omega_0 = B_{1,\epsilon} \). Since \( \|\bm{w}\| \leq 1 \), there are \( (1/\epsilon)^D \) possible values defined over a grid with spacing \(\epsilon\), so the quantum sample can be represented with \( O(D \log(1/\epsilon)) \) qubits. Hence, the quantum sample \( \ket{\pi_t}' \) can be expressed as
$$
    \ket{\pi_t}' = \sum_{\bm w \in \mathcal{P}_{t,\epsilon}} \sqrt{\frac{1}{|\mathcal{P}_{t,\epsilon}|}} \ket{\bm w}.
$$

\begin{theorem}[Hit-and-Run Quantum Walk] \label{theorem_4}
Let \( P_t \) be an ergodic symmetric Markov chain defined by the \textbf{Hit-and-Run} algorithm as in Section~\ref{HCPRW} with stationary uniform distribution \( \bm{\pi}_t \) on the support $\mathcal{P}_{t} $. Suppose we have an initial distribution
$$
    \ket{\pi_t} = \ket{\pi_t'} \ket{\bm 0},
$$
the following properties hold:
\begin{itemize}
    \item \textbf{Mixing time:} \( d_{\mathrm{TV}}(P_t^\tau \cdot \bm\pi_{t+1}, \bm\pi_t) \leq \epsilon \).
    \item \textbf{Warmness:} \( \| \bm\pi_{t+1} / \bm\pi_t \| = O(1) \), where \( \| \bm\pi / \bm\sigma \| := \int_{\mathbb{R}^D} \frac{\pi(\bm x)}{\sigma(\bm x)} \pi(\bm x) \, \mathrm{d}\bm x\).
    \item \textbf{Overlap:} \( |\langle \pi_t | \pi_{t+1} \rangle| = \int_{\mathbb{R}^D} \sqrt{\pi_t(\bm x) \pi_{t+1}(\bm x)} \, d\bm x = \Omega(1) \).
\end{itemize}
Then, we can obtain a state \( \ket{\tilde{\pi}_{t+1}} \) via a quantum-walk-evolution unitary \( \tilde{U}_{t,m} \) such that
$$
    \ket{\tilde{\pi}_{t+1}} = \tilde{U}_{t,m} \ket{\pi_t}
$$
with
$$
    \|\ket{\tilde{\pi}_{t+1}} - \ket{\pi_{t+1}} \| \leq \epsilon,
$$
using \( O\bigl(\sqrt{\tau} \log(1/\epsilon)\bigr) \) calls to the quantum walk operators \( W_t \).
\end{theorem}

\begin{proof}
See Appendix~\ref{secC1}
\end{proof}

\begin{lemma}
 In our QCP-QW algorithm, the overlap condition $\left|\left\langle\pi_t |\pi_{t+1}\right\rangle\right|=\Omega(1)$ is met. 
\end{lemma}
\begin{proof}
 To show the overlap, we notice that as a result of \cite{bertsimas2004solving}, the volume cut-off $\geq1/3$ with high probability in each round, so equivalently the volume left is at least $1/3$. Therefore,
 $$
|\braket{\pi_t|\pi_{t+1}}|=\int_{\mathbb{R}^D}\sqrt{\frac{1}{V(\mathcal{P}_{t})V(\mathcal{P}_{t+1})}} \mathrm{d}\bm x= \frac{ V(\mathcal{P}_{t+1})}{\sqrt{V(\mathcal{P}_{t}) V(\mathcal{P}_{t+1})}}\geq \sqrt{\frac{1}{3}}=\Omega(1).
$$
\end{proof} 
\begin{lemma}
In our  QCP-QW algorithm, the warmness conditions $\left\|{\bm\pi}_{t+1} / \bm\pi_t\right\|=O(1)$ is met.
\end{lemma}
\begin{proof}
$\|\bm\pi_{t+1} / \bm\pi_t\|=\int_{\mathbb{R}^D} \frac{\pi_{t+1}(\bm x)}{\pi_t(\bm x)} \pi_{t+1}(\bm x) \mathrm{d} \bm x= \int_{\mathcal{P}_{t+1}} \frac{ \frac{1}{V(\mathcal{P}_{t+1})}}{\frac{1}{V(\mathcal{P}_{t})}} \frac{1}{V(\mathcal{P}_{t+1})}\mathrm{d}\bm x= \frac{V(\mathcal{P}_{t})}{V(\mathcal{P}_{t+1})} \leq3 =O(1)$.
\end{proof}

From Theorem~\ref{theorem_4}, one can obtain a state $\left|\tilde{\pi}_{t+1}\right\rangle=\tilde{U}_{t,m}\ket{\pi_t}$ with $\|\left|\tilde{\pi}_{t+1}\right\rangle-\left|\pi_{t+1}\right\rangle \| \leq \epsilon_1$ by invoking  $ \{R_{t}, R_{t}^{\dagger}, R_{t+1},  R_{t+1}^{\dagger}\}$ $O(\log \frac{1}{\epsilon_1})$ times, with $O(\sqrt{\tau} \log (1 / \epsilon_1))$ calls to the quantum walk operators. Furthermore, by the result of \cite[Corollary 1]{wocjan2008speedup}, that given $\ket{\pi_0}$, one is able to prepare a state through $ |\tilde{\pi}_r \rangle = {\tilde{U}}_{r-1}\cdots{\tilde{U}}_{0} \ket{\pi_0} $, such that 
$\| |\tilde{\pi}_r  \rangle- |\pi_r  \rangle \| \leq r \epsilon_1 $, by invoking the unitaries from $ \{R_{t\in[r]}, R_{t\in[r]}^{\dagger}\}$ at most $O(r\log \frac{1}{\epsilon_1})$ times, with at most $O(r\sqrt{\tau}\log (1 /\epsilon_1))$ calls of \textsc{controlled}-$W_{t\in[r]}$ gates. Let $r\epsilon_1=\epsilon$, one can obtain $\| |\tilde{\pi}_r  \rangle- |\pi_r  \rangle \| = O(\epsilon) $ with at most $O(r\sqrt{\tau}\log (r / \epsilon))$ calls of \textsc{controlled}-$W_{t\in[r]}$ gates. 

It is known that the Hit-and-Run algorithm has the best-known mixing time bounded by $\tau=O^*(D^3)$ given that a warm start condition, and the distribution space is well-rounded~\citep{lovasz1999hit}. Since we are using a uniform sampling scheme as described in Section~\ref{HCPRW} \textbf{Hit-and-Run}, its stationary distribution is uniform. Assuming that the subroutines required by Algorithm~\ref{Algo_QCPQW} can be implemented as described, after $r$ updates with misclassified examples, one can obtain an approximate uniform distribution $\ket{\tilde\pi_r}\in \mathcal{VS}$ that satisfies $\| |\tilde{\pi}_r  \rangle- |\pi_r  \rangle \| = O(\epsilon) $ with $O^*(r D^{1.5})$ calls of \textsc{controlled}-$W_{t\in[r]}$ gates. Moreover, since $\ket{\pi_r}\in \mathcal{VS}$ is a quantum state, one may directly use it to classify quantum state test data by leveraging techniques such as the \emph{swap test}~\citep{buhrman2001quantum}, where a probability of measurement can estimate the overlap of two quantum states.

\subsubsection{Feasibility}\label{feasibility}

\paragraph{Non-destructive mean estimation}
After preparing a uniform superposition via $\ket{\tilde{\pi}_{t}} = \tilde{U}_{t-1,m} \ket{\tilde{\pi}_{t-1}}$, we need to approximate the centroid $\bm{z}_t$ of the distribution $\tilde{\bm{\pi}}_t$. While measuring multiple copies of $\ket{\tilde{\pi}_t}$ is the most straightforward method, it would collapse the quantum state and require substantial resources. 
To circumvent this, we note that non-destructive amplitude estimation techniques~\citep{harrow2020adaptive,cornelissen2023sublinear} (see Appendix~\ref{secC2}), together with recent results on multivariate mean estimation~\citep{van2021quantum,cornelissen2022near}, provide tools that may be useful for this task.

The former guarantees that $M_1 = O^*(1)$ copies of $\ket{\tilde{\pi}_{t-1}}$ can be recovered with high probability after quantum phase estimation~\citep{brassard2000quantum}, while the latter provides a multivariate mean estimation procedure that can be implemented using $O^*(D)$ applications of the quantum walk unitary $\tilde{U}_{t-1,m}$, yielding an estimate with constant $\ell_2$-norm error.

Taken together, these results suggest that a non-destructive multivariate mean-estimation subroutine $U^{\operatorname{mean}}_t$  may be realizable using $O^*(1)$ copies of $\ket{\tilde{\pi}_{t-1}}$, with each copy undergoing $O^{(D)}$ applications of $\tilde{U}_{t-1,m}$ while being recovered with high probability. Under this interpretation, since each $\tilde{U}_{t-1,m}$ involves $O^*(D^{1.5})$ calls to \textsc{controlled}-$W(P_{t \in [r]})$ gates, the resulting circuit  $U^{\operatorname{mean}}_t$ would require $O^*(D\cdot D^{1.5}) = O^*(D^{2.5})$ calls to \textsc{controlled}-$W(P_{t \in [r]})$ gates.

\paragraph{Non-destructive affine transformation estimation}

To satisfy the well-roundedness condition, it is known that $M_2=O^*(D)$ samples from a convex body $\mathcal{P}_t$ suffice to compute an affine transformation $g_t(\bm{w}) = S_t(\bm{w} - \bm{z}_t)$ that renders the body well-rounded~\citep{rudelson1999random}. Specifically, the empirical mean is defined as $\bm{z}_t := \frac{1}{M_2} \sum_{i=1}^{M_2} \bm{w}_i$, and the affine transformation matrix is $S_t = A_t^{-1/2}$, where the empirical covariance matrix $A_t$ is
$$
A_t = \frac{1}{M_2} \sum_{i=1}^{M_2} (\bm{w}_i - \bm{z}_t)(\bm{w}_i - \bm{z}_t)^T.
$$
Hence, $M_2 = O^*(D)$ copies of $\ket{\tilde{\pi}_{t-1}}$ are sufficient in each round.

As suggested by~\cite{chakrabarti2023quantum,li2022quantum}, one may consider a non-destructive affine transformation estimation procedure. Here, we adopt the complexity analysis for the affine transformation subroutine as stated in these works~\citep{chakrabarti2023quantum,li2022quantum}.

In particular, $M_2 = O^*(D)$ copies are used to estimate the covariance structure and derive the corresponding transformation $S_t$, which can be held in an ancillary quantum register. Then, with $O^*(1)$ applications of the non-destructive amplitude estimation circuit~\citep{harrow2020adaptive,cornelissen2023sublinear}, which consists of the quantum-walk-evolution unitary $\tilde{U}_{t-1,m}$, the affine transformation matrix can be estimated to high accuracy\footnote{We note that this complexity hinges on the implementation details of the moment estimation step.}. Meanwhile, a call of  $\tilde{U}_{t-1,m}$ involves $O^*(D^{1.5})$ calls to \textsc{controlled}-$W(P_{t \in [r]})$ gates. Consequently, the circuit $U^{\operatorname{aff}}_t$ requires $O^*(D\cdot D^{1.5})=O^*(D^{2.5})$ calls to \textsc{controlled}-$W_{t \in [r]}$ gates.

Since the quantum state evolution depends on the updated means and affine transformations for each distribution in $\tilde{\bm{\pi}}_{t \in [r]}$, our algorithm is structured such that the circuits $U^{\operatorname{mean}}_{t+1}$ and $U^{\operatorname{aff}}_{t+1}$ are interleaved between the states $\ket{\tilde{\pi}_t}$ and the quantum-walk-evolution unitary $\tilde{U}_{t,m}$, as illustrated in Figure~\ref{AF}.

\paragraph{Quantum-walk-update unitary}

Although Algorithm~\ref{Algo_QCPQW} is presented with complete quantum subroutines, we note that certain components, such as the quantum-walk-update unitary $U(P_t)$ within $W(P_t)$, can be very difficult to implement in practice, and as far as we know, it has not been made available on current quantum hardware. 

Notably, a proof-of-principle realisation of Szegedy's quantum walk, along with a quantum walk search algorithm, has been proposed for the experimental implementation in trapped ion systems~\citep{dunjko2015quantum}, where the authors give generic constructions of quantum walk search components, and they numerically verified the robustness of the scheme with a small-scale simulation. However, as they use a rank-one Markov chain (all columns of $P$ are identical and each of them is the stationary distribution), the implementation of $U(P)$ is simplified.

Recognising the challenges in implementing Szegedy's quantum walk with complex Markov chains on existing hardware, we frame our results as theoretical upper bounds in an idealised, decoherence-free setting and leave the demonstrations on near-term or fault-tolerant devices to future research. Nevertheless, circuit-level implementations of $U(P_t)$ have been explored for both continuous and discrete cases~\citep{chakrabarti2023quantum,li2022quantum}, 
and here we only provide a brief discussion of how to implement it in the discrete case in Appendix~\ref{secC3}.

Overall, the above discussion outlines a plausible implementation strategy based on existing quantum subroutines; in the following, we analyse the resulting complexity under this interpretation.

\subsubsection{Computational Complexity Analysis}
Here, we analyse the query complexity and arithmetic complexity of the proposed QCP-QW algorithm. From Eqn.\@\eqref{quantum_walk}, it can be seen that one  $W(P_t)$ can be implemented by constant calls of $U(P_t)$. And it is proved by \cite{chakrabarti2023quantum,li2022quantum}, that each $U(P_t)$ can be implemented with $O^*(1)$ calls of the quantum membership oracle $\mathcal{Q}_{\mathcal{P}_t}$ with use of binary search (see Appendix~\ref{secC3}), where every $\mathcal{Q}_{\mathcal{P}_t}$ can be realized with $O^*(\sqrt{t})$ ($t\leq r$) quantum queries of Eqn.\@\eqref{quan_mem}. 

Under the implementation described in the feasibility discussion, both $U^{\operatorname{mean}}_{t}$ and $U^{\operatorname{aff}}_{t}$ would require $O^*(D^{2.5}\sqrt{r})$ queries, substituting $r=O^*(D \log({1}/{\gamma}))$, the overall query complexity can be upper bounded by $O^* (D\log{({1}/{\gamma})} \cdot (\sqrt{N}+ D^{2.5} \sqrt{r}))=O^* (D\log{({1}/{\gamma})} \cdot (\sqrt{N}+ D^{3}\sqrt{\log(1/\gamma)})$. In comparison, the  Hit-and-Run HCP-RW algorithm has a query complexity of $ O^*(D\log{({1}/{\gamma})}\cdot(\sqrt{N} +D^{4.5}\cdot\sqrt{\log (1/\gamma)}$ and the CP-RW algorithm has a query complexity of
$ O^* (D\log{({1}/{\gamma})}\cdot (N + D^5\log{({1}/{\gamma})} )$ (see Section~\ref{HCPRW}).

Additionally, since the quantum walk unitary $\tilde{U}_{t,m}$ is needed on each copy of $\ket{\tilde\pi_{t}}$ inside both the $U^{\operatorname{mean}}_{t+1}$ and $U^{\operatorname{aff}}_{t+1}$, and for performing $\tilde{U}_{t,m}$ in a convex body $\mathcal{P}_{t}$, it has to be well-rounded. Hence, an affine transformation function $g_t$ is needed for each round as $\ket{g_{t}\tilde\pi_{t+1}}=\tilde{U}_{{t},m} \ket{g_{t}\tilde\pi_{t}} $. This matrix–vector product requires $O(D^2)$ arithmetic operations classically and dominates the overall arithmetic complexity of the algorithm. We assume that elementary classical arithmetic operations (such as addition and multiplication) can be implemented by reversible quantum circuits efficiently (e.g. the quantum adders or multipliers) acting on $O^*(D)$ qubits.

It is known that every unitary gate in a logical reversible circuit can be approximated to precision $\epsilon$ using a universal gate set with polylogarithmic overhead in $1/\epsilon$ via the Solovay–Kitaev theorem~\citep{dawson2005solovay}. Under this implementation framework, the arithmetic complexity of our algorithm is dominated by the matrix–vector product, yielding a bound of $O^* (D\log{({1}/{\gamma})} \cdot (D^2\cdot(D^{2.5})))=O^* (D^{5.5}\log(1/\gamma))$. In comparison, both the Hit-and-Run CP-RW and HCP-RW algorithms (see Section~\ref{HCPRW}) has $ O^* (D^7\log{({1}/{\gamma})})$ arithmetic operations.

\subsection{Comparisons of Quantum-Enhanced Perceptron Learning Algorithms} \label{comparison}

To summarise, we provide a Table~\ref{tab1} to compare the query complexities (under the respective implementation assumptions) of the classical and quantum-enhanced algorithms that we examined for the perceptron learning problem\footnote{Except for the version space perceptron, which is not applicable, other algorithms' complexities are analysed under the classical/quantum weak online learning setting.}. As we can see, all the quantum-enhanced algorithms achieve some speedups under idealised quantum computing models. 

\renewcommand*{\thefootnote}{\alph{footnote}}

\begin{table}[ht]
\caption{A summary of algorithms and query complexities} \label{tab1}
\begin{tabular}{ll}
\toprule
Algorithm & Complexity \\
\midrule\midrule
Version Space Perceptron 
  & $O^* ( D^{\frac{1}{2}} \frac{N}{\gamma^D} )$ \\
QVSP
  & $O^* ( D^{\frac{1}{4}} \sqrt{ \frac{N}{\gamma^D} } )$\footnotemark[1] \\
  \midrule
Online Perceptron 
  & $O^*( \frac{1}{\gamma^2} \log(1/\gamma^2) N )$ \\
OQP 
  & $O^* ( \frac{1}{\gamma^2} \log(1/\gamma^2) \sqrt{N})$ \\
  \midrule
CP-RW (Hit-and-Run) 
  & $O^* ( D \log( \tfrac{1}{\gamma} ) (N + D^5 \log ( \tfrac{1}{\gamma} )) )$\footnotemark[2] \\
HCP-RW (\ref{Algo_HCPRW})
  & $O^* ( D\log( \tfrac{1}{\gamma} ) ( \sqrt{N} + D^{4.5} \sqrt{\log( \tfrac{1}{\gamma} )} ) )$ \footnotemark[2]  \\
QCP-QW (\ref{Algo_QCPQW}) 
  & $O^* ( D\log( \tfrac{1}{\gamma} ) ( \sqrt{N}+ D^{3} \sqrt{\log( \tfrac{1}{\gamma} )} ) )$ \footnotemark[3] \\
\botrule
\end{tabular}
\footnotetext[a]{Here we only present the corrected worst-case result of the improved-QVSP of \cite{liao2024quadratic}.}
\footnotetext[b]{For Hit-and-Run CP-RW and HCP-RW introduced in Section~\ref{HCPRW}, we simplify the result using $O(\log^2N)\leq O(D)$. Otherwise, tighter bounds can be achieved by replacing an $O^*(D)$ factor in $O^*(D^5)$ and $O^*(D^{4.5})$ with $\log^2 N.$}
\footnotetext[c]{This complexity is derived under the implementation assumptions discussed in the feasibility section.}
\end{table}

\renewcommand*{\thefootnote}{\arabic{footnote}}

From a statistical perspective, the quantum-enhanced CP algorithms we proposed in this paper (HCP-RW and QCP-QW) offer improved margin dependence compared with the OQP and the QVSP algorithms in the worst case, though it comes at the cost of higher dimensional dependence of $\mathrm{poly} D$ relative to the OQP algorithm. However, it is worth noting that the dominant advantage concerning margin  $O^*(D\log(1/\gamma))$ is offered by classical LP rather than an intrinsic statistical advantage from quantum computation.

In terms of algorithmic design, the OQP utilise Grover’s search only in one instance, sampling a superposition over the training dataset to identify  a misclassified example pair. The HCP-RW improves it by adopting another quantum counting algorithm.  While for the QCP-QW, more quantum subroutines are used than just Grover's search and counting. As a result, a speedup of $O^*(D^{1.5})$  may be achievable compared with the HCP-RW in both query and arithmetic complexities.

\section{Discussion}
\subsection{Perceptron Problem Under Complexity Standpoint}

From the above discussion and Table~\ref{tab1}, we can see that for solving the perceptron problem using quantum methods, one can adopt different algorithms, which have their favourable regimes. For example, QVSP favours a low-rank data matrix, OQP favours a large-margin dataset, and the CP method could perform
well in regimes of small-margin and moderate feature dimension $D$.

From an algorithmic perspective, QVSP achieves a quadratic speedup over Monte-Carlo sampling by using Grover’s search, which is known to be optimal for unstructured search. Nevertheless, this may not be the most effective strategy. More broadly, related lower-bound results in quantum convex optimisation provide useful context. It has been shown that solving LPs in the black-box setting requires at least $\Omega(D\sqrt{N})$ queries~\citep{van2020convex}; while in the white-box model, a tight lower bound of $\Omega(\sqrt{ND})$ queries to the data matrix has also been conjectured~\citep{apers2026quantum}. These results, together with ours, suggest the utility of viewing the version space problem within the context of convex optimisation. And this naturally motivates the open question of identifying the most efficient algorithmic framework for the perceptron problem.

From a complexity-theoretic perspective, whether the version space volume problem could be resolved efficiently using quantum superposition remains an open problem. Especially, if the goal is to reach a version space (or a distribution) without the help of example instances (or separation hyperplanes) that define it, the task resembles a quantum state preparation problem. A direct approach based on unstructured search leveraging Grover's algorithm would typically incur a cost proportional to  $O^*(1/\sqrt{\pi(x)})$, where $\pi(x)$ can be exponentially small when the state space is exponentially large~\citep{wocjan2008speedup,harrow2020adaptive}. A general and efficient solution to such sampling tasks could have significant complexity-theoretic implications, potentially relating to the complexity classes SZK and BQP~\citep{aharonov2003adiabatic,orsucci2018faster}.

From the above discussion, we explicitly separate the two aspects of algorithmic efficiency and fundamental limitation.  We would like to emphasise that our contributions presented in this work only focus on the perspective of algorithm efficiency, rather than the fundamental limitations of quantum computing.

\subsection{Practical Limitations}\label{limitations}
Regarding algorithmic frameworks, we can see that both HCP-RW and QCP-QW rely heavily on Grover’s search and its extension of quantum counting. The fully quantum QCP-QW algorithm additionally requires quantum phase estimation, quantum-walk updates, and non-destructive amplitude estimation, making it technically more demanding. Beyond these technical differences, both approaches share some practical limitations in realising hardware-level speedups.

First, for realistic machine learning tasks, the dataset size $N$ is typically large. Achieving the speedup of $O^*(\sqrt{N})$ over the classical $O^*(N)$ baseline requires an efficient data access scheme (oracle) or QRAM ~\citep{aaronson2015read}. However, in practice, implementing such access can be computationally expensive and may be limited by hardware constraints such as data movement and I/O bandwidth. It has been shown in a perfect fault-tolerant quantum computer\footnote{We refer to the hypothetical fault-tolerant quantum computer with $10^4$ error-corrected logical qubits, $10~\mu s$ gate time for logical operations, all-to-all connectivity, and fully fault-tolerant two-qubit operations as investigated by ~\cite{hoefler2023disentangling}.} that a $10^4$ slowdown is possible in I/O bandwidth~\citep{hoefler2023disentangling}. 

The second issue is that, on real devices, decoherence is unavoidable, and oracles are imperfect.  Both algorithms are based on a strong oracle model (as in Eqn\@\eqref{online_q}). It is well-known that Grover-type speedups degrade significantly in the presence of oracle imperfections~\citep{shenvi2003effects,regev2008impossibility}. Additionally, both algorithms require $O^*(\sqrt{N})$ repetitions of Grover's unitary, which can result in substantial circuit depth and cause decoherence for NISQ devices. 

One may attempt to solve the problem by using a fault-tolerant quantum computing framework. Asymptotically, quantum computation does provide speedup since fewer oracles will be needed on a quantum computer than on a classical computer. However, as highlighted in recent analyses of realistic hardware costs, a constant-factor gap of $\sim 10^{10}$ operation throughput  between a classical GPU and a fault-tolerant quantum computer is representative~\citep {hoefler2023disentangling}.

Hence, the number of oracle calls needs to scale $\sim 10^{10}$ for a quadratic speedup of $O^*(\sqrt{N})$. 
This, in turn, places strong constraints on the allowable complexity of each oracle call. In particular, to achieve quantum advantage within a practical time scale (e.g., two weeks), each oracle would need to involve only $O^*(1)$ arithmetic operations. In our setting, this requirement is highly restrictive, since each oracle evaluation requires at least $\Omega(D)$ operations to compute vector inner products\footnote{Although we assume a black-box oracle model, considering hardware practicality, one needs to open it.}.

Taken together, these observations suggest that our results should be viewed as conceptual, asymptotic improvements, rather than immediately realisable practical speedups on either NISQ or fault-tolerant quantum platforms.

\section{Conclusion and Outlook}\label{conclu} 

In conclusion, through this work, we first point out the limitation of QVSP, which has an exponential dependence on data dimension under worst-case scenarios, and recognise that it may still perform effectively, such as in low-rank regimes. We then introduce two new quantum-enhanced perceptron learning algorithms grounded in the formulation of LPs, and we discuss their practical limitations. As such, we provide a refined understanding of the asymptotic computational complexity behaviours and potentials of some variant of quantum perceptron models. We hope that the present work provides useful insight into both the potential and the limitations of these models.

Furthermore, although current practical constraints are significant (as discussed in Section~\ref{limitations}), they do not preclude the possibility of meaningful quantum advantage in future settings. Such advantages may arise from alternative problem regimes (e.g., structured inputs), improved algorithmic frameworks achieving higher-order speedups, or advances in quantum hardware and system architectures.

Beyond the current work, we note that the classifier produced by the CP-RW method corresponds to a Bayes point machine (BPM) solution~\citep{herbrich2001bayes}, which is known to possess strong theoretical and empirical generalisation guarantees~\citep{minka2001family}. Exploring whether similar properties can be  realised within the proposed quantum-enhanced frameworks, under suitable implementation assumptions, remains an open question for future work.

Although the QCP–QW algorithm relies on a very strong oracle model and should therefore be interpreted primarily as a conceptual upper bound rather than a near-term implementable scheme~\citep{shenvi2003effects,regev2008impossibility}, the approach may provide a conceptual perspective for convex optimisation problems, since standard reductions allow convex minimisation problems to be reformulated as feasibility problems~\citep{nemirovski1994efficient}.  {However, extending the present framework to general optimisation settings would require a more complete implementation of the underlying quantum subroutines, as well as satisfying all feasibility requirements outlined in Section~\ref{feasibility}. This remains an open direction for future work.

\begin{appendices}

\renewcommand{\thefigure}{\arabic{figure}}
\setcounter{figure}{3}

\section{Asymptotic approximation} \label{Appen_A}

Here we provide more detailed analysis regarding the asymptotic approximation  for the incomplete beta function of Lemma~\ref{lemma_2}, where $\Pr [\{\bm{w}|\measuredangle (\bm{w},\bm u)< \arcsin\gamma\}  ]= \frac{1}{2}\mathrm{I}_{\gamma^2}(\frac{D-1}{2},\frac{1}{2}).$

We can validate Lemma~\ref{lemma_2} under boundary conditions first. When $\gamma=0$ (inseparable), $I_0(a,b)=0$, indicating the probability of finding a correct classifier is zero; when $\gamma=1$ (two classes are maximally separated), $I_1(a,b)=\frac{1}{2}$, the probability of finding a correct classifier is $\frac{1}{2}$. These results are consistent with the actual classification situation. 

On the other hand, $ \frac{1}{2}\mathrm{I}_{\gamma^2}(\frac{D-1}{2},\frac{1}{2})=\frac{1}{2}\frac{\mathrm{B}(\gamma^2 ;\frac{D-1}{2},\frac{1}{2})}{\mathrm{B}(\frac{D-1}{2},\frac{1}{2})}$, where $\mathrm{B} (a, b  )=\frac{\Gamma (a  ) \Gamma (b  )}{\Gamma (a+b  )}$ denotes the beta function and $\mathrm{B}(x; a, b)$ denotes the incomplete beta function. Expressing it as a continued fraction expansion, $\mathrm{B}(x;a,b)=\frac{x^a(1-x)^b}{a (1+\frac{d_1}{1+} \frac{d_2}{1+} \frac{d_3}{1+} \frac{d_4}{1+} \cdots  )}$, with coefficients $d_{2 m+1}=-\frac{(a+m)(a+b+m) x}{(a+2 m)(a+2 m+1)}$, and $d_{2 m}=\frac{m(b-m) x}{(a+2 m-1)(a+2 m)}$. Thus, in the limit of $\gamma  \rightarrow 0$,
\begin{equation*}
\begin{aligned}
    \mathrm{B}(\gamma^2 ;\frac{D-1}{2},\frac{1}{2})
    &=\frac{2\gamma^{D-1}\sqrt{(1-\gamma^2)}}{(D-1) (1-O(\gamma^2)  )} \sim \frac{2\gamma^{D-1}}{(D-1)}.
\end{aligned}
\end{equation*}
When $D\geq 2$ is a constant, one has $\mathrm{B}(\frac{D-1}{2},\frac{1}{2}) \leq \pi$, by property of gamma function. Accordingly, 
\begin{equation}
    \Pr [\{\bm{w}|\measuredangle (\bm{w},\bm u)< \arcsin\gamma\}  ] \sim\frac{\gamma^{D-1}}{(D-1)}/ {\mathrm{B}(\frac{D-1}{2},\frac{1}{2})} \geq \frac{\gamma^{D-1}}{\pi (D-1)}.
\end{equation}
In the limit of  $D \rightarrow\infty$, one can apply Stirling's approximation $\mathrm{B}(a, b) \sim \Gamma(b) a^{-b}$, such that
\begin{equation}
    \begin{aligned}
    &\Pr [\{\bm{w}|\measuredangle (\bm{w},\bm u)< \arcsin\gamma\}  ] \sim\frac{\gamma^{D-1}}{(D-1)}/ \sqrt{\frac{2\pi}{{D-1}}} 
    \sim \frac{\gamma^{D-1}}{\sqrt{2\pi (D-1)}}.
    \end{aligned}
\end{equation}

\section{Classical and Quantum Markov Chains}\label{Markov_chian}
Let \( P = (p_{xy})_{\bm x,\bm y \in \Omega} \) denote the stochastic transition matrix of a Markov chain with finite state space \(\Omega\), where \(|\Omega| = N\). The matrix element \( p_{xy} \) represents the transition probability from state \( \bm x \) to \( \bm y \), satisfying \( \sum_y p_{xy} = 1 \) for all \( \bm x \in \Omega \). If $P$ is \emph{ergodic}, which implies irreducibility, the Perron–Frobenius theorem  guarantees a unique eigenvalue-1 eigenvector \( \bm{\pi} \) such that \( \bm{\pi}^T P = \bm{\pi}^T \). The vector \( \bm{\pi} = (\pi_{\bm x})_{\bm x \in \Omega} \) is called the stationary distribution. \emph{Time-reversibility} means \( P = P^* \), while \emph{symmetry} means \( P = P^T \); in the symmetric case, the stationary distribution is uniform.

The spectral gap of \( P \) is defined as \( \delta(P) = 1 - |\lambda_1| \), where the eigenvalues satisfy \( 1 = \lambda_0 > |\lambda_1| \geq \ldots \geq |\lambda_{N-1}| \geq 0 \). For a classical Markov chain, the mixing time \( \tau(P) \) of an initial distribution \( \bm{\rho} \) is defined by the condition \( d_{TV}(P^\tau \bm{\rho}, \bm{\pi}) \leq \epsilon \) for some \( 1 > \epsilon > 0 \), where \( d_{TV} \) is the total variation distance. It is known that the mixing time is upper bounded by the inverse of the spectral gap, i.e., \( \tau(P) = O(1/\delta) \).

The quantum analogue of the spectral gap, called the phase gap, is defined as
$$
\Delta(P) = 2 \theta_1,\quad \text{with }
\cos \theta_j = |\lambda_j|, \quad j\in[0,\cdots,N-1]
$$
The phase gap represents the minimal angular distance of 1 to the other eigenvalues on the complex plane, and satisfies
$$
\Delta(P) \geq 2 \sqrt{\delta(P)}.
$$

\section{Quantum speedup} \label{secC}
\subsection{Proof of Theorem~\ref{theorem_4}}\label{secC1}
For proving Theorem~\ref{theorem_4}, we first need the following Lemmas~\ref{effective_gap},~\ref{Grover_fix_point},~\ref{implement_R}. 

\begin{lemma}[Effective spectral gap for warm start] {~\cite[Proposition 4.2]{chakrabarti2023quantum}, ~\cite[Lemma 5]{li2022quantum}}\label{effective_gap}
For an ergodic time-reversible Markov chain $P_t$, if the state $\ket{{\pi}_{t+1}}$ is a warm start of the stationary sates $\ket{{\pi}_t}$ which mixes up to a $d_{TV}\leq\epsilon$ after $\tau$ steps, then $\left|{\pi}_{t+1}\right\rangle=\left|\pi_{t+1, \text { good }}\right\rangle+\left|e_{t+1}\right\rangle$, where $\left|\pi_{t+1, \text { good }}\right\rangle$ lies in the subspace spanned by the fast-mixing eigenvectors $\left|\psi_j^\pm\right\rangle$ of $W(P_t)$, such that its eigenvalue $e^{\pm2 \pi i \theta_j}$ satisfies $\theta_j=0$ or $\theta_j =\Omega(\sqrt{1/\tau})$. And the state $\ket{e_{t+1}}$ is spanned by those slow-mixing eigenvectors whose  $0<\theta_i<O(\sqrt{1/\tau})$, has 
 its $\ell_2$ norm bounded by $\|\left|e_{t+1}\right\rangle \| \leq \sqrt{\epsilon}$.

\end{lemma}

\begin{lemma}[Grover's $\frac{\pi}{3}$ fix point search]{~\cite[Lemma 1]{wocjan2008speedup}}\label{Grover_fix_point}
Let $\ket{\pi_t}$, $\ket{\pi_{t+1}}$ be two arbitrary quantum states with overlap $|\braket{\pi_t|\pi_{t+1}}|^2\geq p$ for some $0<p\leq 1$ is a constant. Given $\ket{\pi_t}$, one is able to prepare a state through $ |\tilde{\pi}_{t+1} \rangle = {U}_{t,m} \ket{\pi_t} $, such that 
$\| |\tilde{\pi}_{t+1}  \rangle- |\pi_{t+1}  \rangle \| \leq \epsilon_1$ for any $\epsilon_1 >0$, by choosing $m=O(\log(\log {1}/{\epsilon_1}))$, which invokes the unitaries from $ \{R_{t}, R_{t}^{\dagger}, R_{t+1},  R_{t+1}^{\dagger}\}$ at most $M=O(\log {1}/{\epsilon_1})$ times.  Here the form of unitaries $U_{t,m}$ is defined recursively, such that $U_{t,0}=\mathbb{1}$, and $U_{t,i+1}=U_{t,i}\cdot R_t \cdot U_{t,i}^\dagger \cdot R_{t+1} \cdot U_{t,i}$. The unitaries $R_{t}$, $R_{t+1}$ operators are selective phase shifts that $R_{t}=e^{\frac{i \pi}{3}}\operatorname{\Pi}_t+\operatorname{\Pi}_t^\perp$, where $\operatorname{\Pi}_t$ is the orthogonal projector onto $\operatorname{span}\{\ket{\pi_t}\}$. 
\end{lemma}
\begin{proof}
The state $\ket{\pi_t}$ is a uniform superposition of states on the convex body $\mathcal{P}_t$. It can be decomposed in the subspace $\mathcal{H}_{\pi_t}=\operatorname{span}\{\ket{\pi_{t+1}},\ket{\pi_{t+1}^\perp}\}$ as 
$\ket{\pi_t}= \sin\theta_{t+1} \ket{\pi_{t+1}} +\cos\theta_{t+1} \ket{\pi_{t+1}^\perp}$, where 
$$
|\pi_{t+1}\rangle=\frac{\operatorname\Pi_{\mathcal{P}_{t+1}}|\pi_{t}\rangle}{\| \operatorname\Pi_{\mathcal{P}_{t+1}}|\pi_{t}\rangle \|}= \frac{\sum_{\bm x \in \mathcal{P}_{t+1}} \ket{\bm x} \langle \bm x|\sum_{\bm x' \in \mathcal{P}_{t}} \sqrt{\pi_{t,\bm x'}}|\bm x'\rangle |\bm p_x'   \rangle} { \sqrt{\bra{\pi_{t}} \operatorname\Pi_{\mathcal{P}_{t+1}}|\pi_{t}\rangle} }=\frac{\sum_{\bm x \in \mathcal{P}_{t+1}} \sqrt{\pi_{t,\bm x}}|\bm x\rangle |\bm p_x  \rangle}{\sqrt{\sum_{\bm x \in \mathcal{P}_{t+1}} \pi_{t,\bm x}}},
$$
and its orthogonal complement 
$$
|\pi_{t+1}^\perp\rangle=\frac{(\mathbb{1}-\operatorname\Pi_{\mathcal{P}_{t+1}}) |\pi_t\rangle}{\| (\mathbb{1}-\operatorname\Pi_{\mathcal{P}_{t+1}})|\pi_t\rangle \|}= \frac{\sum_{\bm x \notin \mathcal{P}_{t+1}} \ket{\bm x} \langle \bm x|\sum_{\bm x' \in \mathcal{P}_t} \sqrt{\pi_{t,\bm x'}}|\bm x'\rangle |\bm p_x'   \rangle} { \sqrt{\bra{\pi_t} (\mathbb{1}-\operatorname\Pi_{\mathcal{P}_{t+1}})|\pi_t\rangle} }=\frac{\sum_{\bm x \notin \mathcal{P}_{t+1}} \sqrt{\pi_{t,\bm x}}|\bm x\rangle |\bm p_x  \rangle}{\sqrt{\sum_{\bm x \notin \mathcal{P}_{t+1}} \pi_{t,\bm x}}}.
$$
A transition from $\ket{\pi_t}$ to $\ket{\pi_{t+1}}$ can be realized by the Grover's $\pi/3$ amplitude amplification algorithm~\citep{grover2005different}, such that there is a rotation $U_{t,m}$ yields $ |\langle \pi_{t+1}| U_{t,m}| \pi_t\rangle|^ 2 \geq 1-(1-p)^{M}$. From \cite[Corollary 1]{wocjan2008speedup}, to make sure $\| |\tilde{\pi}_{t+1}  \rangle- |\pi_{t+1}  \rangle \| \leq \epsilon_1$, one needs  $M=O(\frac{\log \frac{1}{\epsilon_1}}{\log(\frac{1}{1-p})}
)$ uses of  $R_{t}$, $R_{t+1}$ and theirs inverse. 
\end{proof}

\begin{lemma}\label{implement_R}
Let $P_t$ be an ergodic and symmetric Markov chain, the state $\ket{{\pi}_{t+1}}$ be a warm start of the stationary uniform state $\ket{{\pi}_t}$ which mixes up to a $d_{TV}\leq\epsilon_3$ after $\tau$ steps.  To implement $R_{t+1}$, one needs one quantum query as defined in Eqn.\@\eqref{online_q} with one ancillary qubit. To approximately implement $R_t$ up to an amplitude error $O(\sqrt{\epsilon_3})$, a number of ancillary qubits $ac=O(\log \sqrt{\tau} \log(1/\sqrt{\epsilon_3}))$ suffice with $O(\sqrt{\tau}\log\frac{1} {\sqrt{\epsilon_3}})$ invocations of \textsc{controlled}-$W_t$ gates.
\end{lemma}

\begin{proof}
It can be easily verified that if we have a unitary $P_{t+1}$, that $P_{t+1}\ket{\pi_{t+1}} \ket{0} =\ket{\pi_{t+1}} \ket{0}$ and $P_{t+1}\ket{\pi_{t+1}^\perp} \ket{0}  =\ket{\pi_{t+1}^\perp} \ket{1}$, then $R_{t+1}$ can be constructed by letting 
\begin{equation}\label{R_{t+1}}
    R_{t+1}=P_{t+1}^\dagger (\mathbb{1} \otimes(e^{\frac{i \pi}{3}}\ket{0}\bra{0}+\ket{1}\bra{1})  ) P_{t+1}
\end{equation}
as shown by \cite{wocjan2008speedup}. From the action of $P_{t+1}$, one can see that $P_{t+1}$ is effectively a quantum query that is similar to the Eqn.\@\eqref{online_q}, such 
\begin{equation*}
P_{t+1}\ket{\bm w}\ket{\bm x'_{t+1}y'_{t+1},\bm z_t}\ket{0}=\ket{\bm w}\ket{\bm x'_{t+1}y'_{t+1},\bm z_t}\ket{0\oplus f_{\bm w}(\bm x'_{t+1}y'_{t+1},\bm z_{t})},
\end{equation*}
where $f_{\bm w} (\bm x'_{t+1}y'_{t+1},\bm z_{t}) = 1 $ iff $(\bm w-\bm z_{t})^T \bm x'_{t+1} y'_{t+1}\leq 0$ ($\bm w\in\mathcal{P}_{t+1}$). Therefore, $R_{t+1}$ can be implemented simply with one quantum query and one ancillary qubit. 

The realisation of $R_{t}$ is more computationally expensive.  Here, we adopt the quantum walk search scheme proposed by \cite{magniez2007search}. The phase estimation~\citep{cleve1998quantum} circuit $PE(W_t)$, composed of \textsc{controlled}-$W_t$ gates, can be leveraged to realise unitaries as $R_{t}$. It is further analysed in \cite[Corollary 2]{wocjan2008speedup}, that for some $\epsilon_2>0$, 
there is a quantum circuit $V_t$ that acts on $\mathbb{C}^{N} \otimes\mathbb{C}^{N} \otimes (\mathbb{C}^2  )^{\otimes a c}$. It uses the $PE(W_t)$ circuit $c$ times with $ac$ ancillary qubits, where $a=O(\log \frac{1}{\Delta})$, $c=O(\log \frac{1}{\sqrt{\epsilon_2}})$). A demonstration of the $PE(W_t)$ circuit is shown in Figure~\ref{PE} below.  

\begin{figure}[ht]
\centering
\scalebox{0.9}{
\begin{quantikz}[row sep=0.5cm, column sep=0.5cm]
\lstick [wires=4]{$\ket{0}^{\otimes a}$ }
&   & \gate{H} &\qw & \qw      & \phantom{a}\cdots \phantom{a}   & \ctrl{4} & \gate[wires=4]{QFT^{-1}_{2^a}} & \qw \\
&  & \gate{H} & \qw      & \ctrl{3} &\phantom{a}\cdots \phantom{a}       & \qw      &                                & \qw \\
&\lstick{\vdots} \setwiretype{n} &\setwiretype{n} &\setwiretype{n}&  &\phantom{a}
\cdots \phantom{a}   &\setwiretype{n}  \\
&   & \gate{H} &  \ctrl{1}   & \qw      &\phantom{a}\cdots \phantom{a}    & \qw &                                & \qw \\
\lstick{$\ket{\psi_k}$   } &\qw & \qw & \gate{W^{2^0}_t} &\gate{W^{2^1}_t} & \phantom{a}\cdots \phantom{a}  & \gate{W^{2^{a-1}}_t}  & \qw & \qw 
\end{quantikz}}
    \caption{Phase estimation circuit $PE(W_t)$, where $H$ is the Hadamard gate, $QFT^{-1}$ denotes the inverse quantum Fourier transform circuit, and $W_t=U(P_t)^{\dagger} S U(P_t) R_{\mathcal{A}} U(P_t)^{\dagger} S U(P_t) R_{\mathcal{A}}$}
    \label{PE}
\end{figure}

The realisation of $V_t$ invokes the \textsc{controlled}-$W_t$ gates at most $O(\frac{1}{\Delta}\log (1 / \sqrt{\epsilon_2}))$ times, where $\Delta=2\theta_1$ as introduced in Section~\ref{Szegedy} is the phase gap. The $V_t$  circuit has the following properties, 
\begin{equation}
\begin{aligned}
V_t |\pi_t  \rangle |0\rangle^{\otimes a c} & = |\pi_t  \rangle |0\rangle^{\otimes a c}, \\
V_t |\psi_k^\pm \rangle |0\rangle^{\otimes a c} & =\sqrt{1-\epsilon_2} |\psi_k^\pm  \rangle |\chi_k  \rangle+\sqrt{\epsilon_2} |\psi_k  ^\pm \rangle|0\rangle^{\otimes a c},
\end{aligned}
\end{equation}
where $ \langle 0^{\otimes ac}|\chi_k  \rangle=0$. The eigenstates $|\pi_t\rangle$ and $ 
\ket{\psi_k^\pm}$ lie in the subspace $\mathcal{H}_{\pi_t}$, with $W(P_t)|\pi_t\rangle =|\pi_t\rangle$ and $W(P_t)\ket{\psi_k^{\pm}}= e^{\pm 2\pi i\theta_k} \ket{\psi_k^\pm}$ for all $k\in [M]$. Essentially, the circuit $V_t$ differentiates the stationary state  $|\pi_t\rangle$ and its orthogonal eigenstates $\ket{\psi_k^\pm}$  up to some error $O(\sqrt{\epsilon_2})$. Similarly to $R_{t+1}$, one can construct a circuit  

\begin{equation}
    \tilde{R_t}=V_t^\dagger (\mathbb{1} \otimes (e^{\frac{i \pi}{3}} \ket{0}\bra{0}^{\otimes ac}+(\mathbb{1}-\ket{0}\bra{0}^{\otimes ac})  )  ) V_t,
\end{equation} 
that $\| \tilde{R}_t|\psi_k^\pm\rangle \ket{0}^{\otimes ac}-{R}_t|\psi_k^\pm\rangle \ket{0}^{\otimes ac} \| \leq 2 \sqrt{\epsilon_2}$.

From Lemma~\ref{Grover_fix_point}, it can be seen that a transition from $\ket{\pi_t}$ to $\ket{\tilde{\pi}_{t+1}}$ can be achieved by a sequence of $\tilde{U}_{t,m}=\cdots R_{t+1}\tilde{R}_{t}R_{t+1}\ket{\pi_t}$, and the state after applying the first $R_{t+1}$ as in Eqn.\@\eqref{R_{t+1}} will be in the space of $\operatorname{span}\{\ket{\pi_{t+1}},\ket{\pi_{t+1}^\perp}\}$. Then for implementing $\tilde{R}_{t}$ composed of the phase estimation circuit $PE(W_t)$, one need to make sure  that $\ket{{\pi}_{t+1}}$ has a large overlap with the good state  $\ket{\pi_{t+1, \text { good}}}$, which is spanned by  $\ket{\psi_j^\pm}$ with $\theta_j=\Omega(\sqrt{1/\tau})=\Omega(\Delta')$.  This is guaranteed by Lemma~\ref{effective_gap} as a warm start condition is satisfied. Therefore, for implementing $\tilde{R}_{t}$, it is sufficient to implement  a circuit $V_t$ that effectively differentiates states  $\ket{\psi_j^\pm}$ and  ${\ket{\pi_t}}$.  Furthermore, as  proved in Lemma~\ref{effective_gap}, $\|\left|e_{t+1}\right\rangle \| \leq \sqrt{\epsilon_3}$,  one can implement $\tilde{R}_t$ such that $\| \tilde{R}_t\left|\phi\right\rangle-{R}_t\left|\phi\right\rangle \| \leq 2 \sqrt{\epsilon_3}$ for any state $\ket{\phi}$ appears in the Grover's ${\pi}/{3}$ rotation from $\ket{{\pi}_{t+1}}$ to $\ket{{\pi}_{t}}$. Therefore,  it takes $ac=O(c \log \frac{1}{\Delta'})=O(c \log \sqrt{\tau})$ ancillary to implement the circuit $V_t$  and it  invokes the \textsc{controlled}-$W_t$ gate at most $O(\frac{1}{\Delta'}\log (1 / \sqrt{\epsilon_3}))= O(\sqrt{\tau}\log (1 / \sqrt{\epsilon_3})) $ times.
\end{proof}

\begin{proof}
From the above Lemmas~\ref{effective_gap},~\ref{Grover_fix_point},~\ref{implement_R}, one can see that the total error of preparing state $\ket{\tilde{\pi}_{t+1}}$ through $ |\tilde{\pi}_{t+1} \rangle = \tilde{U}_{t+1,m} \ket{\pi_t} $ is bounded by $O(\epsilon_1+(\frac{1}{p} \log \frac{1}{\epsilon_1}) \sqrt{\epsilon_3})$. Taking $\epsilon_1=O(\epsilon)$, $\epsilon_3=O(\epsilon^2)$, one can obtain  $\|\left|\tilde{\pi}_{t+1}\right\rangle-\left|\pi_{t+1}\right\rangle \|= O(\epsilon)$ by using $O(\sqrt{\tau} \log(1 /{\epsilon}))$ calls to the quantum walk operators. Theorem~\ref{theorem_4} is then proved.
\end{proof}

\subsection{Non-destructive Amplitude Estimation}\label{secC2}

\begin{theorem}{~\cite[Theorem 18]{harrow2020adaptive}}\label{NAE}
 Given state $|\psi\rangle$ and reflections $R_\psi=2|\psi\rangle\langle\psi|-$ $I$ and $R=2 P-I$, and any $\eta>0$, there exists a quantum algorithm that outputs $\bar{a}$, an approximation to $a=\langle\psi| P|\psi\rangle$, so that

\begin{equation}
    |\bar{a}-a| \leq 2 \pi \frac{a(1-a)}{M}+\frac{\pi^2}{M^2}
\end{equation}
with probability at least $1-\eta$ and $O(\log (1 / \eta) M)$ uses of $R_\psi$ and $R$. Moreover, the algorithm restores the state $|\psi\rangle$ with probability at least $1-\eta$.
\end{theorem}

It takes use of the  amplitude estimation circuit  $(F_M^{-1} \otimes I) \Lambda_M(Q)(F_M \otimes I)$ proposed by \cite{brassard2000quantum} (as shown in Figure~\ref{AE}) with $Q=-R_{\psi}R$. Here the output $\tilde{a}=\sin ^2(\pi \frac{y}{M})$ is an approximation of $a=\braket{\psi|P|\psi}$, and the $\bar{a}$ in Theorem~\ref{NAE} is the median of multi-copy estimations of $\tilde{a}$.

\begin{figure}[ht]
    \centering
    \scalebox{0.9}{
    \begin{tikzpicture}

\def\qubitsep{1.0}
\def\colsep{1}
\foreach \i in {0,2} {
\draw [thick]  (0, -\i*\qubitsep) -- (6*\colsep, -\i*\qubitsep);}
\node at (-0.2*\colsep, -1*\qubitsep) {\vdots};
\draw [thick] (0, -3.5*\qubitsep) -- (6*\colsep, -3.5*\qubitsep);
\node[left] at (0, 0) {$\ket{0}$};
\node[left] at (0, -2*\qubitsep) {$\ket{0}$};
\node[left] at (0, -3.5*\qubitsep) {$\ket{\psi}$};
\draw[fill=white, thick] (0.8*\colsep, 0.4*\qubitsep) rectangle (2.2*\colsep, -2.4*\qubitsep);
\node at (1.5*\colsep, -1*\qubitsep) {$QFT_M$};
\draw[fill=white,thick] (3.8*\colsep, 0.4*\qubitsep) rectangle (5.2*\colsep, -2.4*\qubitsep);
\node at (4.5*\colsep, -1*\qubitsep) {$QFT_M^{-1}$};
\draw[fill=white,thick] (2.5*\colsep, -4*\qubitsep) rectangle (3.5*\colsep, -3*\qubitsep);
\node at (3*\colsep, -3.5*\qubitsep) {$Q^j$};
\draw[fill=black] (3*\colsep, 0) circle (2pt);
\node  at (3*\colsep, -1*\qubitsep) {\vdots};
\draw[fill=black] (3*\colsep, -2*\qubitsep) circle (2pt);
\draw [thick] (3*\colsep, 0)-- (3*\colsep, -0.7*\qubitsep);
\draw [thick] (3*\colsep, -3*\qubitsep)-- (3*\colsep, -1.5*\qubitsep);
\node[right] at (6*\colsep, -1*\qubitsep) {$\ket{y}$};
\end{tikzpicture}
}
\caption{Amplitude estimation circuit~\citep{brassard2000quantum}, where $F_M$ is the quantum Fourier transform, the \textsc{controlled}--${Q}^j$ operator in the middle is $\Lambda_M(Q)=\sum_{j=0}^{M-1}|j\rangle\langle j| \otimes {Q}^j$}
\label{AE}
\end{figure}

Fitting in our QCP-QW algorithm, at the $t^{th}$ round, one has the initial state $\ket{{\tilde\pi}_t}$ and the aim is to estimating the mean $\bm z_{t}$ and affine transformation matrix $S_{t}$ of distribution $\ket{\tilde{\pi}_{t}}=\ket{\psi}$. The unitary $Q$ can be applied by letting  $Q=-R_{\tilde{\pi}_t} R=\tilde{U}_{t-1,m}(\mathbb{1}-2\ket{{\tilde\pi}_{t-1}}\bra{{\tilde\pi}_{t-1}}) \tilde{U}_{t-1,m}^\dagger R = (\mathbb{1}-2\ket{{\tilde\pi}_{t}}\bra{{\tilde\pi}_{t}}) R$. It can be easily seen that here the operation $\mathbb{1}-2\ket{{\tilde\pi}_{t-1}}\bra{{\tilde\pi}_{t-1}}$ can be realized similarly to the selective phase shift defined in Lemma~\ref{Grover_fix_point}, by letting $R'_{t-1}=e^{i \pi}\operatorname{\Pi}_{t-1}+\operatorname{\Pi}_{t-1}^\perp=\mathbb{1}-2{\Pi}_{t-1}$. From Lemma~\ref{implement_R}, it can be shown that the $R'_{t-1}$ gate can be implemented with $O^*(D^{1.5})$ calls to the \textsc{controlled}-$W_{t-1}$ gates. Moreover, it is proved in Theorem~\ref{theorem_4}, the quantum-walk-evolution unitary  $\tilde{U}_{t-1,m}$ can be realised using $O^*(\sqrt{\tau})=O^*(D^{1.5})$ calls to the \textsc{controlled}-$W_{t-1}$ gates. Hence, the unitary $R_{\tilde\pi_t}$  can be implemented with $O^*(D^{1.5})$ calls to the \textsc{controlled}-$W_{t-1}$ gates. We note that the observable $P$ needs to be designed carefully for different observables. 
To summarise, on each copy of quantum state $\ket{\psi}=\ket{\tilde{\pi}_t}=  \tilde{U}_{t-1,m} \ket{\tilde{\pi}_{t-1}}$, one is allowed to estimate a component of the covariance matrix $A_t$ up to some error and restore the sample with high probability with $O^*(1)$ uses of $R_{\tilde\pi_t}$ and $P$. As a result, on each copy, the non-destructive estimation circuit takes $O^*(D^{1.5})$ calls of the \textsc{controlled}-$W_{t-1}$ gates.

\subsection{Quantum-Walk-Update Operator Implementation}\label{secC3}

\begin{lemma}[Transition probability of the Hit-and-Run walk
]{~\cite[Lemma 3]{lovasz1999hit}}
For the \textbf{Hit-and-Run} with uniform sampling scheme defined as in Section~\ref{HCPRW}. If the current point of Hit-and-Run is $\bm x$, then the density function of the distribution of the next point $\bm y$ is
\begin{equation}\label{HAR_prob}
f_{\bm x}(\bm y)=\frac{2}{D \pi_D} \frac{1}{\ell(\bm x, \bm y)\|\bm x-\bm y\|^{D-1}},
\end{equation}
where $\pi_D=\frac{\pi^{D / 2}}{\Gamma(1+D / 2)}$ is the volume $B_1$  unit ball, and $\ell(\bm x, \bm y)$ is the chord length through $\bm x$ and $\bm y$.
\end{lemma}

\begin{lemma}
The hit-and-run quantum-walk-update unitary 
$U(P_t)$ can be achieved with $O^*(1)$ uses of membership oracles  $\mathcal{O}_{\mathcal{P}_t}$ and $O^*(D)$ gate complexity, in both discrete and continuous quantum computers. 
\end{lemma}

\begin{proof}
    The implementation of the hit-and-run quantum-walk-update operator has been illustrated in detail for both continuous and discrete cases by \cite{chakrabarti2023quantum, li2022quantum}. In the discrete case (see Eqn.\@\eqref{quantum_update}), 
$$
U(P_t)|\bm x\rangle|\bm 0\rangle=\ket{\bm x}\ket{\bm p_x}=\ket{\bm x}\sum_{\bm y \in \mathcal{P}_{t,\epsilon}} \sqrt{p_{xy}}\ket{\bm y},
$$
and in the continuous case, 
$$
U(P_t)|\bm x\rangle|\bm0\rangle=\int_{\mathcal{P}_t} \sqrt{p_{xy}}|\bm x\rangle|\bm y\rangle \mathrm{d} \bm y,
$$
with $p_{xy}$ denoting the $x,y$ element of the stochastic transition matrix $P_t$. Here we only describe briefly how the unitary can be constructed in the discrete case, and interested readers can refer to the paper~\citep{chakrabarti2023quantum} for a construction in continuous quantum computing. As we mentioned in Section~\ref{QCP_QW}, we take the $\epsilon$-$descritiztion$ representation, such that our quantum states are defined in the $\epsilon$-$net$ of $\mathcal{P}_{t,\epsilon} \in\mathbb{R}^D$, thus $O^*(D)$ qubits are needed for representing quantum samples in the space. 

At the $t^{th}$ round, the walk starts from a given current point $\bm x$ that  belongs to the \emph{current position register} $\in\mathcal{H}_{curr}$. For choosing a direction uniformly, one can construct a $D$-dimensional Gaussian distribution to obtain a uniform distribution of all directions on the $S^{D-1}$ sphere. This can be done using $O^*(D)$ ancillary \emph{direction registers} $\in \mathcal{H}_{dir}$, and since the Gaussian distribution can be efficiently integrable, one can leverage the Grover-Rudolph method~\citep{grover2002creating} to prepare the states. 
Then, for a direction $\bm u$, the line $\ell=\bm x+ t\bm u$ uniquely defines a chord of $\mathcal{P}_t\cap \ell$ with its length $\ell(t_1,t_2)$, where $t\in [t_1,t_2] $ is a scalar and  $t_1,t_2$ are the two endpoints. It is known that classically, the endpoints and the chord length can be determined within an error $\epsilon$ by using $O(\log(1/\epsilon))$ membership oracles $\mathcal{O}_{\mathcal{P}_t}$. Then, for the next proposed point, a new state $\ket{\bm y}=\frac{1}{\sqrt{|t_2-t_1|}}\sum_{t\in[t_1,t_2]} \ket{\bm x+t \bm u}$ of a \emph{new position register}  $\in \mathcal{H}_{new}$ can be constructed with controlled operations on $\bm x \in\mathcal{H}_{curr}$. Uncomputing the direction registers, the desired state of $\ket{\bm x}\sum_{\bm y \in \mathcal{P}_{t,\epsilon}} \sqrt{p_{xy}}\ket{\bm y}$ as in Eqn.\@\eqref{quantum_update},  can be acquired. After change of variables, it can be verified that the probability $p_{xy}$ of moving from $\bm x$ to $\bm y$ is the same as $f_{\bm x}(\bm y)$ of Eqn.\@\eqref{HAR_prob}.
Overall, it is evident that the hit-and-run quantum-walk-update unitary
$U(P_t)$ can be achieved with $O^*(1)$ uses of membership oracles and $O^*(D)$ gate complexity. 

\end{proof}





\end{appendices}
\section*{Declarations}
\begin{itemize}
\item Funding - This work is supported by the Plan France 2030 through the PEPR integrated project EPiQ ANR-22-PETQ-0007; by the ANR JCJC DisQC ANR-22-CE47-0002-01; by the QuanTEdu-France ANR-22-CMAS-0001, and as part of the Initiative d’Excellence d’Aix-Marseille Université—A*MIDEX AMX-21-RID-011.
\item Conflict of interest - The authors have no competing interests to declare that are relevant to the content of this article.
\item 
Ethics approval and consent to participate - Not applicable.
\item Consent for publication - All authors consent to the publication of this
work.
\item Data availability - The data of Figure~\ref{hard_data} is publicly
available as referenced in the text. Except that, no datasets were generated or analysed during the current study.

\item Materials availability - Not applicable.
\item Code availability - Not applicable.
\item Authors' contribution - All authors contributed to the conception and design of the study.  Xiaoyu Sun and Mathieu Roget performed the formal analyses of Section 3. Xiaoyu Sun investigated Section 4.3 and drafted the initial manuscript. Hachem Kadri and Giuseppe Di Molfetta proposed substantial amendments during the revision process. Hachem Kadri and Giuseppe Di Molfetta proposed and supervised the project. All authors read and commented on previous versions of the manuscript and approved the final manuscript.
\end{itemize}

\bibliography{sn-bibliography}

\end{document}